\documentclass[twocolumn,10pt,journal]{IEEEtran}
\usepackage{indentfirst,amsmath}
\usepackage{color,rotating,pdflscape}
\usepackage{cite}
\usepackage{amsfonts,amsmath,amssymb,amsthm}
\usepackage{latexsym,amscd}
\usepackage{amsbsy,extarrows}
\usepackage{graphicx,multirow,bm}
\usepackage[table]{xcolor}
\usepackage{extarrows}
\usepackage{tikz}
\usetikzlibrary{arrows,shapes}
\usetikzlibrary{shadows}

\newtheorem{Theorem}{Theorem}

\newtheorem{Remark}{Remark}

\newtheorem{Proposition}{Proposition}

\begin{document}
\title{Private and Secure Distributed  Matrix Multiplication Schemes for Replicated or MDS-Coded Servers}

\author{Jie Li, \IEEEmembership{Member,~IEEE}, and Camilla Hollanti, \IEEEmembership{Member,~IEEE}
\thanks{Partial results from this paper (mainly related to the code $\mathcal{C}_{\rm PSDMM}$ from Section \ref{sec:PSDMM}) was presented at the 2021 IEEE International Symposium on Information Theory (ISIT) \cite{li2021improved}. The work of J. Li was supported in part by the National Science Foundation of China under Grant No. 61801176. The work of C. Hollanti was supported by the Academy of Finland, under Grants No. 318937 and 336005.}
\thanks{J. Li was with the Department of Mathematics and Systems Analysis,
                    Aalto University, FI-00076 Aalto,  Finland, and also with the Hubei Key Laboratory of Applied Mathematics, Faculty of Mathematics and Statistics, Hubei University,
Wuhan 430062, China (e-mail: jieli873@gmail.com).}
\thanks{C. Hollanti is with the Department of Mathematics and Systems Analysis,
                    Aalto University, FI-00076 Aalto,  Finland (e-mail:  camilla.hollanti@aalto.fi).}
}
\date{}
\maketitle

\begin{abstract}
In this paper, we study the problem of \emph{private and secure distributed matrix multiplication (PSDMM)}, where a  user having a private matrix $A$  and $N$ non-colluding servers sharing a library of $L$ ($L>1$) matrices $B^{(0)}, B^{(1)},\ldots,B^{(L-1)}$, for which the user  wishes to compute $AB^{(\theta)}$ for some $\theta\in [0, L)$ without revealing any information of  the matrix $A$ to the servers, and keeping the index $\theta$ private to the servers. 
Previous work is limited to the case that the shared library (\textit{i.e.,} the matrices $B^{(0)}, B^{(1)},\ldots,B^{(L-1)}$) is stored across the servers in a replicated form and schemes are very scarce in the literature,   there is still much room for improvement.  
In this paper, we propose two PSDMM schemes, where one is limited to the case that the shared library is stored across the servers in a replicated form but has a better performance than  state-of-the-art schemes in that it can achieve a smaller recovery threshold and download cost. The other one focuses on the case that the shared library is stored across the servers in an MDS-coded form, which requires less storage in the servers. The second PSDMM code does not subsume the first one even if the underlying MDS code is degraded to a repetition code as they are totally two different schemes. 
\end{abstract}


\section{Introduction}

\IEEEPARstart{M}{atrix} multiplication is  one of the key operations in many science and engineering applications, such as machine learning and cloud computing. Carrying  out the computation on  distributed servers are desirable for improving efficiency and reducing the
user's computation load, as the user can divide the computation at hand into several sub-tasks to be carried out by the helper
servers. As a downside, when scaling out computations across many distributed servers,
 the computation latency can be affected by
orders of magnitude due to the presence of stragglers, see \emph{e.g.}, \cite{joshi2017efficient,wang2015using}.
Fortunately, recent works
have shown that coding techniques can reduce
the computation latency  \cite{lee2017speeding,yu2017polynomial,dutta2018unified,dutta2019optimal,jia2021cross,yu2020straggler}.

As the  computations are scaling out across   many distributed servers, besides stragglers, security is also a concern as   the servers  might be curious about the matrix contents.
To this end,  consider the problem where the user has two matrices $A$ and $B$ and wishes to compute their product with the assistance of $N$ distributed servers, which are  honest but curious in that any $X$ of them may collude to deduce information about either $A$ or $B$ \cite{d2020gasp,d2021degree}.
This raises the problem of SDMM,
which has recently received a lot of attention from an information-theoretic perspective   \cite{chang2018capacity,d2021degree,kakar2019capacity,yang2019secure,jia2021capacity,d2020gasp,aliasgari2020private,yu2020entangled,zhu2021secure,zhu2021improved}.

In another variant of this problem,
user privacy  should also be taken into account. Consider the scenario where the  user has a private matrix $A$  and there are $N$ non-colluding servers sharing a library of $L$ matrices $B^{(0)}, B^{(1)},\ldots,B^{(L-1)}$, for which the user  wishes to compute $AB^{(\theta)}$ for some $\theta\in [0, L)$ without revealing any information of  the matrix $A$ to the servers and keeping the index $\theta$ of the desired matrix $B^{(\theta)}$  hidden from the servers. This problem is termed PSDMM, 
which is also quite relevant to the problem of private information retrieval (PIR) \cite{chor1995private,sun2017capacity,sun2017capacityrobust,freij2017private,banawan2018capacity,tajeddine2018private,freij2018t,kumar2019achieving,zhu2019new,chang2019upload,wang2019symmetric,d2019one,zhou2020capacity,jia2020x,li2020towards,holzbaur2021toward}.

Generally, four performance metrics are of particular interest for PSDMM schemes and more generally for  any matrix multiplication scheme:
\begin{itemize}
  \item The upload cost:  the amount of data transmitted from the user to the severs to assign the sub-tasks;
 \item Server storage cost: the amount of data stored in each server;
  \item The download cost: the amount of data to be downloaded from the servers;
  \item The recovery threshold $R_c$: the  number of servers that need to complete their tasks before the user can recover the desired matrix product(s).
\end{itemize}
The goal is to design PSDMM schemes that can minimize one or several of the above metrics.

Up to now, the study of PSDMM is still very scarce in the literature. In \cite{chang2019upload}, a PSDMM code  based on PIR was proposed, which can provide a flexible tradeoff between the upload cost and download cost by adjusting the partitioning sizes of the matrices, but requires a high sub-packetization degree and cannot mitigate   stragglers. In  \cite{kim2019privateCL},   PSDMM codes based on the polynomial codes in \cite{yu2017polynomial} were presented.  In \cite{aliasgari2020private}, Aliasgari \textit{et al.} proposed a new   PSDMM code  based on the entangled polynomial codes in \cite{yu2020straggler}, which generalizes the one in  \cite{kim2019privateCL} and allows for a flexible tradeoff between the upload cost and the download cost (or equivalently,  recovery threshold). Very recently,    Yu \emph{et al.} \cite{yu2020entangled}  studied the problem of PSDMM based on  Lagrange coded computing with bilinear complexity \cite{blaser2013fast}.  The detailed performance of the above codes  together with the new codes proposed in this work will be illustrated in the following sections.

In this paper, we propose two PSDMM codes $\mathcal{C}_{\rm PSDMM}$ and $\mathcal{C'}_{\rm PSDMM}$ with the shared library being stored across the servers in a replicated form and in an MDS-coded form, respectively. 
The proposed codes have the following advantages:
\begin{itemize}
  \item The new code $\mathcal{C}_{\rm PSDMM}$ outperforms  the ones in \cite{aliasgari2020private} and \cite{kim2019privateCL} in that it has a smaller recovery threshold as well as  download cost under the same upload cost, and it provides a more flexible tradeoff between the upload and download costs than the code in \cite{kim2019privateCL}.
 The new code $\mathcal{C}_{\rm PSDMM}$  also has a smaller recovery threshold and  download cost compared to the codes in \cite{yu2020entangled} and \cite{chang2019upload}, for some parameter regions. In addition,  it  can tolerate stragglers and does not require sub-packetization, hence being superior to the code in \cite{chang2019upload}.

\item The new PSDMM code $\mathcal{C'}_{\rm PSDMM}$ with the shared library being stored across the servers in an MDS-coded form, can greatly reduce the storage overhead in the servers. To the best of our knowledge, this is the first time to consider PSDMM from MDS-coded servers.
\end{itemize}
Note that the PSDMM code $\mathcal{C'}_{\rm PSDMM}$ does not subsume $\mathcal{C}_{\rm PSDMM}$ even if the underlying MDS code is degraded to a repetition code as they are totally two different schemes. More detailed comparisons and analysis will be given in the following sections.

\subsection{Organization}
The rest of this paper is organized as follows.  Section II introduces the problem settings and summarizes the main results. In Section \ref{sec:PSDMM}, a new efficient PSDMM code from replicated servers is proposed, followed by detailed comparisons with previous work. Section \ref{sec:PSDMM2} presents a PSDMM code from MDS-coded servers.  Finally, Section \ref{sec:conclusion} draws the concluding remarks.

\section{Problem settings and main results}

Throughout this paper, we assume that the matrix $A$ is a $t\times s$ matrix over a certain  finite field $\mathbf{F}$, and $B^{(i)}$, $i\in [0, L)$,  are  $s\times r$ matrices over $\mathbf{F}$, where $L> 1$, $t, s$, and $r$ are some positive integers. Let $m, p$, and $n$ be some positive integers such that  $m\mid t$, $p\mid s$, $n\mid r$.

\subsection{Problem settings}
In the following, we introduce the problem setting for PSDMM.

In PSDMM, the  user has a private matrix $A$  and there are $N$ non-colluding servers sharing a library of $L$ ($L>1$) matrices $B^{(0)}, B^{(1)},\ldots,B^{(L-1)}$ in a replicated form or an MDS-coded form, for which the user  wishes to compute $AB^{(\theta)}$ for some $\theta\in [0, L)$ without revealing any information of  the matrix $A$ to the servers, and keeping the index $\theta$ private to the servers. 

Typically, the PSDMM schemes considered in this paper contain the following three phases.
\begin{itemize}
\item Encoding phase: The user encodes $A$  to obtain $\tilde{A}_i$  for $i\in [0, N)$.
\item Query, communication and computation phase (replicated servers): The user sends $\tilde{A}_i$ and a query $q_i^{(\theta)}$ (which is usually independent
of the matrix $A$) to server $i$, who then first encodes the library into $\tilde{B}_i$ based on the received query,   and then computes $Y_i=\tilde{A}_i\tilde{B}_i$ and returns the result $Y_i$ to the user, where $i\in [0, N)$.
\item Query, communication and computation phase (MDS-coded servers): The user sends a query $q_i^{(\theta)}$ (which contains $\tilde{A}_i$) to server $i$, who then performs some calculation and returns the result $Y_i$ to the user, where $i\in [0, N)$.
\item Decoding phase: From the results returned by the fastest $R_c$ servers, the user can then decode the desired matrix multiplication $AB^{(\theta)}$.
\end{itemize}

We say that an encoding scheme for PSDMM is \emph{private}  if the index $\theta$ is kept secret from any individual server,  \textit{i.e.,}
\begin{equation*}
  I(\theta;q_i^{(\theta)}, \tilde{A}_i, \mathbf{B}_i)=0,
\end{equation*}
for any $i\in [0, N)$ and $\theta\in [0, L)$, where $\mathbf{B}_i$ denotes the data stored in server $i$, \textit{e.g.,} $\mathbf{B}_i$ stands for $B^{(0)}, B^{(1)},\ldots,B^{(L-1)}$ for PSDMM schemes for replicated servers. The scheme is secure if    no  information is leaked about the matrix $A$ to any server, \textit{i.e.,}
\begin{equation*}
  I(q_i^{(\theta)},\tilde{A}_i; A)=0,
\end{equation*}
for any $i\in [0, N)$, where $I(X; Y)$ denotes the mutual information between $X$ and $Y$.

The upload cost is defined as the number of elements in $\mathbf{F}$ transmitted from the user to the servers, \textit{i.e.,}  $\sum\limits_{i=0}^{N-1}(|\tilde{A}_i|+|q_i^{(\theta)}|)$ for replicated servers and $\sum\limits_{i=0}^{N-1}|q_i^{(\theta)}|$ for MDS-coded servers, while the download cost is defined
as the  number of elements in $\mathbf{F}$ that the user downloaded from the fastest $R_c$ servers $j_0,\ldots, j_{R_c-1}$, maximized over $\{j_0,\ldots, j_{R_c-1}\}\subseteq [0, N)$,   \textit{i.e.,} $$\max_{\{j_0,\ldots, j_{R_c-1}\}\subseteq [0, N)}\sum\limits_{i=j_0}^{j_{R_c-1}}|Y_i|,$$ where $|Y|$ denotes the size of the matrix $Y$ counted as $\mathbf{F}$ symbols.

\subsection{Main results}
This subsection includes the main results derived in this paper. The proofs are carried out in the later sections.
\begin{Theorem}\label{Thm_C5}
There exists an explicit PSDMM code $\mathcal{C}_{\rm PSDMM}$ from replicated servers, with the upload cost  $N\frac{ts}{mp}$,  download cost  $R_c\frac{tr}{mn}$, recovery threshold
\begin{equation*}
R_c=pmn+pm+n,
\end{equation*}
and each server stores $Lsr$ elements from $\mathbf{F}$.
\end{Theorem}

\begin{Theorem}\label{Thm_C6}
There exists an explicit PSDMM code $\mathcal{C}'_{\rm PSDMM}$ from MDS-coded servers,  with each server stores $L\frac{sr}{pn}$ elements from $\mathbf{F}$,  the upload cost $LN\frac{ts}{mp}$,  download cost  $R_c\frac{tr}{mn}$, and recovery threshold
\begin{equation*}
R_c=pmn+pn-1
\end{equation*}
if the underlying field size is at least $N\ge R_c$, or $R_c=pmn+n+p-1$ if the underlying finite field is sufficiently large.
\end{Theorem}

\section{Private and Secure  Distributed Matrix Multiplication From Replicated Servers}\label{sec:PSDMM}
In this section, we propose a new PSDMM code from replicated servers. The  construction is similar to the one in \cite{aliasgari2020private}, but  the encoding phase of the new PSDMM scheme is more efficient. This leads to a smaller recovery threshold and download cost. Before presenting the general construction, we first give a motivating example.

\subsection{A motivating example}\label{sec:psdmm-ex}
Assume that the user possesses a matrix $A\in \mathbf{F}^{t\times s}$ and there are $L=2$ matrices $B^{(0)}, B^{(1)}\in \mathbf{F}^{s\times r}$ stored across the $N$ servers in a replicated form, \textit{i.e.,} each server holds the matrices $B^{(0)}, B^{(1)}$, where $t,s,r$ are even. Suppose the user wishes to compute $AB^{(0)}$.
Partition the matrices $A$ and $B^{(i)}$ into block matrices
\begin{eqnarray*}
 A=\begin{bmatrix}
       A_{0,0} & A_{0,1} \\
       A_{1,0} & A_{1,1}  \\
     \end{bmatrix},\,\, B^{(i)}=\begin{bmatrix}
       B_{0,0}^{(i)} & B_{0,1}^{(i)} \\
       B_{1,0}^{(i)} & B_{1,1}^{(i)} \\
     \end{bmatrix},
\end{eqnarray*}
where $A_{k,j}\in \mathbf{F}^{\frac{t}{2}\times \frac{s}{2}}$ and $B_{j', k'}^{(i)}\in \mathbf{F}^{\frac{s}{2}\times \frac{r}{2}}$. Then
\begin{equation*}
  AB^{(0)} =\begin{bmatrix}
       C_{0,0}^{(0)} & C_{0,1}^{(0)}  \\
       C_{1,0}^{(0)} & C_{1,1}^{(0)}  \\
     \end{bmatrix}
\end{equation*}
where
\begin{equation}\label{Eqn_entry_C_ex2}
C_{i,j}^{(0)}=A_{i,0}B_{0,j}^{(0)}+A_{i,1}B_{1,j}^{(0)},~ i,~j\in [0,~2).
\end{equation}

Let $Z$ be a random matrix over $\mathbf{F}^{\frac{t}{2}\times \frac{s}{2}}$.  Then, define a  polynomial
\begin{equation}\label{Eqn_PS_ex_f}
 f(x)\hspace{-.7mm} =\hspace{-.7mm} A_{0,0}x^{\alpha_{0,0}}\hspace{-.7mm}+\hspace{-.7mm}A_{0,1}x^{\alpha_{0,1}}\hspace{-.7mm}+\hspace{-.7mm}A_{1,0}x^{\alpha_{1,0}}\hspace{-.7mm}+\hspace{-.7mm}A_{1,1}x^{\alpha_{1,1}}\hspace{-.7mm}+\hspace{-.7mm} Zx^{\gamma}
  \end{equation}
where $\alpha_{k,j}, \gamma$ are some integers to be specified later.

Let $a_{1}$ and $a_{0,0}, \ldots, a_{0,N-1}$ be $N+1$ pairwise distinct elements from  $\mathbf{F}$.
For every $i\in [0, N)$, the  user first evaluates $f(x)$ at $a_{0, i}$, then sends $f(a_{0, i})$  and the query
\begin{equation}\label{Eqn_query_ex}
 q_i^{(0)}=(a_{0,i},a_{1})
\end{equation}
to server $i$. Upon receiving the query $q_i^{(0)}$, server $i$   encodes the library into a matrix as $g(a_{0, i})$ where $g(x)$ is defined as
\begin{equation}\label{Eqn_g_ex_C1}
   g(x)=\sum\limits_{j=0}^{1}\sum\limits_{k=0}^{1}B_{j,k}^{(0)}x^{\beta_{j,k}}+\sum\limits_{j=0}^{1}\sum\limits_{k=0}^{1}B_{j,k}^{(1)}a_{1}^{\beta_{j,k}}.
\end{equation}
After encoding the library, server $i$ computes $f(a_{0,i})g(a_{0,i})$ and then returns the result to the user.

Let $h(x)=f(x)g(x)$, the expression is shown in \eqref{Eqn-h-private-ex} in the next page,
\begin{figure*}
\hrulefill
\begin{align}
\nonumber  h(x)
=&\underbrace{(A_{0,0}B_{0,0}^{(0)}x^{\alpha_{0,0}+\beta_{0,0}}+A_{0,1}B_{1,0}^{(0)}x^{\alpha_{0,1}+\beta_{1,0}}) +(A_{0,0}B_{0,1}^{(0)}x^{\alpha_{0,0}+\beta_{0,1}}+A_{0,1}B_{1,1}^{(0)}x^{\alpha_{0,1}+\beta_{1,1}})}_{\mbox{useful\  terms}}\\
\nonumber&+\underbrace{(A_{1,0}B_{0,0}^{(0)}x^{\alpha_{1,0}+\beta_{0,0}}+A_{1,1}B_{1,0}^{(0)}x^{\alpha_{1,1}+\beta_{1,0}}) +(A_{1,0}B_{0,1}^{(0)}x^{\alpha_{1,0}+\beta_{0,1}}+A_{1,1}B_{1,1}^{(0)}x^{\alpha_{1,1}+\beta_{1,1}})}_{\mbox{useful\  terms}}\\
 \nonumber&+\underbrace{\sum\limits_{k=0}^{1}\sum\limits_{k'=0}^{1}(A_{k,0}B_{1,k'}^{(0)}x^{\alpha_{k,0}+\beta_{1,k'}}+A_{k,1}B_{0,k'}^{(0)}x^{\alpha_{k,1}+\beta_{0,k'}})+
  \sum\limits_{j'=0}^{1}\sum\limits_{k'=0}^{1}ZB_{j',k'}^{(0)}x^{\gamma+\beta_{j',k'}}}_{\mbox{interference\  terms}}\\
  \label{Eqn-h-private-ex}
&+ \underbrace{ \left(A_{0,0}x^{\alpha_{0,0}}+A_{0,1}x^{\alpha_{0,1}}+A_{1,0}x^{\alpha_{1,0}}+A_{1,1}x^{\alpha_{1,1}}+ Zx^{\gamma}\right) \left(B_{0,0}^{(1)}a_1^{\beta_{0,0}}+B_{0,1}^{(1)}a_1^{\beta_{0,1}}+B_{1,0}^{(1)}a_1^{\beta_{1,0}}+B_{1,1}^{(1)}a_1^{\beta_{1,1}}\right)}_{\mbox{interference\  terms}},
\end{align}
\end{figure*}
then the results returned from the servers are exactly the evaluations of $h(x)$ at some evaluation points.

The user wishes to obtain the data in \eqref{Eqn_entry_C_ex2} (related to the useful terms in $h(x)$) from any $R_c$ out of  the $N$ evaluations of $h(x)$, which can be fulfilled if the following conditions hold.

\begin{itemize}
  \item [(i)] For $k\in [0, 2)$ and $k'\in [0, 2)$,
\begin{eqnarray*}
  \alpha_{k,0}+\beta_{0,k'} = \alpha_{k,1}+\beta_{1,k'}.
\end{eqnarray*}
  \item [(ii)]For $U=\{\alpha_{k,0}+\beta_{0,k'}|0\le k<2, 0\le k'<2\}$ and
  \begin{eqnarray*}
   I&=&\{\alpha_{k,j}+\beta_{j',k'}|0\le k, k'<2, 0\le j\ne j'<2\}\\
   &&\cup  \{\gamma+\beta_{j',k'}|0\le j',k'<2\}\\
   &&\cup\{\alpha_{k,j}|0\le k, j<2\}\cup \{\gamma\},
  \end{eqnarray*}
  $|U|=4$ and $U\cap I= \emptyset$.
  \item [(iii)] $R_c=\deg(h(x))+1$.
\end{itemize}

The task can be finished because
\begin{itemize}
  \item (i) guarantees that each $C_{k,k'}^{(0)}$ appears in $h(x)$,
  \item (ii) guarantees that   $C_{k,k'}^{(0)}, k,k'=0,1$ are coefficients of different terms of  $h(x)$ with different degrees, which are different from the   degrees of the interference terms of $h(x)$, \textit{i.e.,} each $C_{k,k'}^{(0)}$ is the coefficient of a unique term in $h(x)$,
  \item (iii) guarantees the decodability from Lagrange interpolation \cite{stoer2013introduction}.
  \item By \eqref{Eqn_PS_ex_f} and  (ii), we can get $\gamma\ne \alpha_{k,j}$ for $k,j=0, 1$, thus one easily obtains $I(\{f(a_i)\}; A)=0$
for any $i\in [0, N)$ and $a_i\in \mathbf{F}$ as $f(a_i)$ is masked by the random matrix $Z$, and $I(q_i^{(0)}; A)=0$ as $q_i^{(0)}$ and $A$ are independent, then security is fulfilled. While the privacy condition is met by the definition of the query vector in \eqref{Eqn_query_ex} for the desired index, the detailed proof  is similar to \cite{kim2019privateCL}.
\end{itemize}

We provide a concrete exponent assignment  for this example in Table \ref{Table_private_ex}. From the given assignment, we see that $\deg(h(x))=13$ and thus $R_c=14$.

\begin{table}[htbp]
\begin{center}
\caption{An assignment for exponents of $f(x)$ in \eqref{Eqn_PS_ex_f} and $g(x)$ in \eqref{Eqn_g_ex_C1}, where $\spadesuit,\heartsuit,\clubsuit,\diamondsuit$ are used to highlight the exponents of the useful terms of $h(x)$ in \eqref{Eqn-h-private-ex}, \textit{i.e.,} all the elements in $U$}\label{Table_private_ex}
\begin{tabular}{|c|c|c|c|c|c|c|c}
\hline $+$ & $\beta_{0,0}=4$ & $\beta_{0,1}=9$ & $\beta_{1,0}=3$& $\beta_{1,1}=8$ \\
\hline $\alpha_{0,0}=1$ & $5\ \ \spadesuit$ & $10\ \ \heartsuit$ & $4$& $9$ \\
\hline $\alpha_{0,1}=2$& $6$ & $11$ & $5\ \ \spadesuit$& $10\ \ \heartsuit$ \\
\hline $\alpha_{1,0}=3$& $7\ \ \clubsuit$ & $12\ \ \diamondsuit$ & $6$& $11$  \\
\hline $\alpha_{1,1}=4$& $8$ & $13$ & $7\ \ \clubsuit$& $12\ \ \diamondsuit$  \\
\hline $\gamma=0$& $4$ & $9$ & $3$& $8$  \\
  \hline
\end{tabular}
\end{center}
\end{table}

\subsection{General construction}\label{sec:PSDMM-ge-con}
In the following, we propose a general construction for PSDMM from replicated servers.
Partition the matrices $A$ and $B^{(i)}$ ($i\in [0, L)$) into block matrices as
\begin{equation*}
A=[A_{k,j}]_{k\in [0,m), j\in [0, p)},
   \,\, B^{(i)}=[B_{j',k'}^{(i)}]_{j'\in [0,p), k'\in [0, n)},    
\end{equation*}
where $A_{k,j}\in \mathbf{F}^{\frac{t}{m}\times \frac{s}{p}}$ and $B_{j', k'}^{(i)}\in \mathbf{F}^{\frac{s}{p}\times \frac{r}{n}}$.

Then
\begin{equation*}
  AB^{(i)} =[C_{k,k'}^{(i)}]_{k\in [0,m), k'\in [0, n)},
\end{equation*}
where
\begin{equation*}\label{Eqn_entry_AB-i}
  C_{k,k'}^{(i)} = \sum\limits_{j=0}^{p-1} A_{k,j}B_{j,k'}^{(i)}, \, k\in [0, m),\ \ k'\in [0,n).
\end{equation*}

Let $Z$ be a random matrix over $\mathbf{F}^{\frac{t}{m}\times \frac{s}{p}}$,  and define a polynomial
\begin{equation}\label{Eqn_geC-fx}
  f(x) = \sum\limits_{k=0}^{m-1}\sum\limits_{j=0}^{p-1}A_{k,j}x^{\alpha_{k,j}}+ Zx^{\gamma},
  \end{equation}
where $\alpha_{k,j}, \gamma$ are some integers to be specified later.

Suppose that the user wishes to compute $AB^{(\theta)}$ for some $\theta\in [0, L)$, then let $a_{0},\ldots,a_{\theta-1},a_{\theta+1},\ldots, a_{L-1}$ and $a_{\theta,0}, \ldots, a_{\theta,N-1}$ be $N+L-1$ pairwise distinct elements chosen from  $\mathbf{F}$.
For every $i\in [0, N)$, the  user first evaluates $f(x)$ at $a_{\theta, i}$, then sends $f(a_{\theta, i})$  and the query
\begin{equation}\label{Eqn_query}
 q_i^{(\theta)}=(a_{0},\ldots,a_{\theta-1},a_{\theta,i},a_{\theta+1},\ldots, a_{L-1})
\end{equation}
to server $i$.

Let
\begin{equation}\label{Eqn_gt}
g_t(x)=\sum\limits_{j=0}^{p-1}\sum\limits_{k=0}^{n-1}B_{j,k}^{(t)}x^{\beta_{j,k}},    
\end{equation}
where $\beta_{j,k}$, $j\in [0, p)$, $k\in [0, n)$ are some integers to be specified later. Let $q_i^{(\theta)}[t]$ be the $t$-th element in the vector $q_i^{(\theta)}$. Then, upon receiving the query $q_i^{(\theta)}$, server $i$ first encodes the library into the following matrix
$$\sum\limits_{t=0}^{L-1}g_t(q_i^{(\theta)}[t])=\sum\limits_{j=0}^{p-1}\sum\limits_{k=0}^{n-1}B_{j,k}^{(\theta)}a_{\theta, i}^{\beta_{j,k}}+\hspace{-1.5mm}\sum\limits_{t=0, t\ne \theta}^{L-1}\sum\limits_{j=0}^{p-1}\sum\limits_{k=0}^{n-1}B_{j,k}^{(t)}a_{t}^{\beta_{j,k}}.$$
For convenience, let
\begin{align*}
   g(x)&=g_{\theta}(x)+\sum\limits_{t=0, t\ne \theta}^{L-1}g_t(q_i^{(\theta)}[t])\\&=\sum\limits_{j=0}^{p-1}\sum\limits_{k=0}^{n-1}B_{j,k}^{(\theta)}x^{\beta_{j,k}}+\underbrace{\sum\limits_{t=0, t\ne \theta}^{L-1}\sum\limits_{j=0}^{p-1}\sum\limits_{k=0}^{n-1}B_{j,k}^{(t)}a_{t}^{\beta_{j,k}}}_{\rm constant}.
\end{align*}
Now it is obvious that $g(a_{\theta, i})=\sum\limits_{t=0}^{L-1}g_t(q_i^{(\theta)}[t])$.

After encoding the library, server $i$ computes $f(a_{\theta, i})g(a_{\theta, i})$ and then sends the result back to the user.

Let $h(x)=f(x)g(x)$, \textit{i.e.,}
{\small
\begin{align}
 \nonumber  h(x)&=\sum\limits_{k=0}^{m-1}\sum\limits_{j=0}^{p-1}\sum\limits_{j'=0}^{p-1}\sum\limits_{k'=0}^{n-1}A_{k,j}B_{j',k'}^{(\theta)}x^{\alpha_{k,j}+\beta_{j',k'}}\\
\nonumber&+ \sum\limits_{j'=0}^{p-1}\sum\limits_{k'=0}^{n-1}ZB_{j',k'}^{(\theta)}x^{\gamma+\beta_{j',k'}}\\
\nonumber &+ \left(\sum\limits_{k=0}^{m-1}\sum\limits_{j=0}^{p-1}A_{k,j}x^{\alpha_{k,j}}+ Zx^{\gamma}\hspace{-.5mm} \right)\hspace{-2mm} \left(\sum\limits_{t=0, t\ne \theta}^{L-1}\sum\limits_{j=0}^{p-1}\sum\limits_{k=0}^{n-1}B_{j,k}^{(t)}a_t^{\beta_{j,k}}\hspace{-.5mm} \right),
\end{align}}then the results returned from the servers are exactly the evaluations of $h(x)$ at some evaluation points,
thus we can derive the following result.

\begin{table*}[htbp]
\begin{center}
\caption{A comparison of some key parameters between the proposed PSDMM code $\mathcal{C}_{\rm PSDMM}$ and some existing ones, where $R(p, m, n)$ denotes the \textit{bilinear complexity} of multiplying an $m\times p$ matrix and a $p\times n$ matrix,which is defined as the minimum number of element-wise multiplications
required to complete such an operation \cite{yu2020straggler}.}
\label{comp_PSDMM}
\begin{tabular}{|c|c|c|c|c|}
\hline  & Upload cost & Download cost & Recovery threshold $R_c$& References \\
\hline Kim--Lee code  & $N\frac{ts}{m}$ & $R_c\frac{tr}{mn}$ & $(m+1)(n+1)$ & \cite{kim2019privateCL}\\
\hline Aliasgari \textit{et al.}  PSGPD code  & $N\frac{ts}{mp}$ & $R_c\frac{tr}{mn}$ & $pmn+pn+pm-p+2$  & \cite{aliasgari2020private}\\
\hline Chang--Tandon code  & $R_c\frac{ts}{m}$ & $tr\frac{m+1}{m}\left(1+\frac{m+1}{R_c}+\cdots+(\frac{m+1}{R_c})^{L-1}\right)$ & $\ge m+1$ & \cite{chang2019upload}\\
\hline Yu \textit{et al.}     PSDMM code    & $N\frac{ts}{mp}$ & $R_c\frac{tr}{mn}$ & $2R(p,m,n)+1$ & \cite{yu2020entangled} \\
\hline The new code $\mathcal{C}_{\rm PSDMM}$ & $N\frac{ts}{mp}$ & $R_c\frac{tr}{mn}$ & $pmn+pm+n$ & Theorem \ref{Thm_C5}\\
\hline
\end{tabular}
\end{center}
\end{table*}

\begin{Proposition}\label{Thm_PSDMM}
For PSDMM schemes for replicated servers, the multiplication of   $A$ and $B^{(\theta)}$ can be securely computed  with the upload cost  $N\frac{ts}{mp}$, download cost  $R_c\frac{tr}{mn}$ with $R_c$ being the recovery threshold,
 if the following conditions hold.
\begin{itemize}
  \item [(i)] For $k\in [0, m)$ and $k'\in [0, n)$,
\begin{eqnarray*}
  \alpha_{k,0}+\beta_{0,k'} = \cdots=\alpha_{k,p-1}+\beta_{p-1,k'}.
\end{eqnarray*}
  \item [(ii)] The set $U=\{\alpha_{k,0}+\beta_{0,k'}|0\le k<m, 0\le k'<n\}$ containing all the exponents of the useful terms of $h(x)$ is disjoint with the set
{\small   
  \begin{align*}
I&=\{\alpha_{k,j}+\beta_{j',k'}|0\le k<m, 0\le j\ne j'<p, 0\le k'<n\}\\
   &\cup \{\gamma+\beta_{j',k'}|0\le j'<p, 0\le k'<n\}\\&\cup\{\alpha_{k,j}|0\le k<m, 0\le j<p\}\cup \{\gamma\}
  \end{align*}
}containing all the exponents of the interference terms of $h(x)$,
  and $|U|=mn$.
  \item [(iii)] $R_c= \deg(h(x))+1=\max(U\cup I)+1$.
\end{itemize}
\end{Proposition}

\begin{proof}
By (ii), one can deduce that $\alpha_{0,0},\ldots,\alpha_{m-1,p-1},\gamma$ are pairwise distinct, then
$I(f(a_{\theta,i}); A)=0$ for any $i\in [0, N)$ and $a_i\in \mathbf{F}$, together with the fact that $q_i^{(\theta)}$ and $A$ are independent, we conclude that security is guaranteed. The privacy condition follows from three conditions: 1) the definition of the query in Eq. \eqref{Eqn_query} for the desired index $\theta\in [0, L)$, 2) the matrices $B^{(0)}, B^{(1)},\ldots,B^{(L-1)}$ are independent of $\theta$, and 3) the matrix $f(a_{\theta, i})$ is random to each server as it is masked by the random matrix $Z$ according to \eqref{Eqn_geC-fx}. The detailed and rigorous proof is similar to \cite{kim2019privateCL}, thus we omit it here. 
By (i),   each $C_{k,k'}$ appears in $h(x)$,
 while (ii) guarantees that   each $C_{k,k'}$ is the coefficient of a unique term in $h(x)$, and finally
  (iii) guarantees the decodability.

It is obvious  that the upload cost   is $\sum\limits_{i=0}^{N-1}|f(a_i)|=N\frac{ts}{mp}$\footnote{Similar to that in \cite{chang2019upload}, we ignore the upload cost for $q_i^{(\theta)}$ as $|q_i^{(\theta)}|\ll|f(a_i)|$.} and the download cost is   $R_c\frac{tr}{mn}$. This finishes the proof.
\end{proof}

In the following, we provide an assignment method for the exponents of $f(x)$ and $g(x)$ (or $g_t(x)$).

\begin{table*}[htbp]
\begin{center}
\caption{A comparison of the normalized upload cost and normalized download cost between $\mathcal{C}_{\rm PSDMM}$ and Chang--Tandon code.}
\label{comp_PSDMM_nor}
\begin{tabular}{|c|c|c|c|c|}
\hline  & Normalized upload cost & Normalized download cost & Recovery threshold $R_c$& References \\
\hline Chang--Tandon code  & $R_c/m$ & $\frac{m+1}{m}\left(1+\frac{m+1}{R_c}+\cdots+(\frac{m+1}{R_c})^{L-1}\right)$ & $\ge m+1$ & \cite{chang2019upload}\\
\hline The new code $\mathcal{C}_{\rm PSDMM}$ & $N/mp$ & $R_c/mn$ & $pmn+pm+n$ & Theorem \ref{Thm_C5}\\
\hline
\end{tabular}
\end{center}
\end{table*}

\begin{figure*}[htbp]
\centering
\hspace{-25mm}
\begin{minipage}[t]{0.4\textwidth}
\includegraphics[scale=.7]{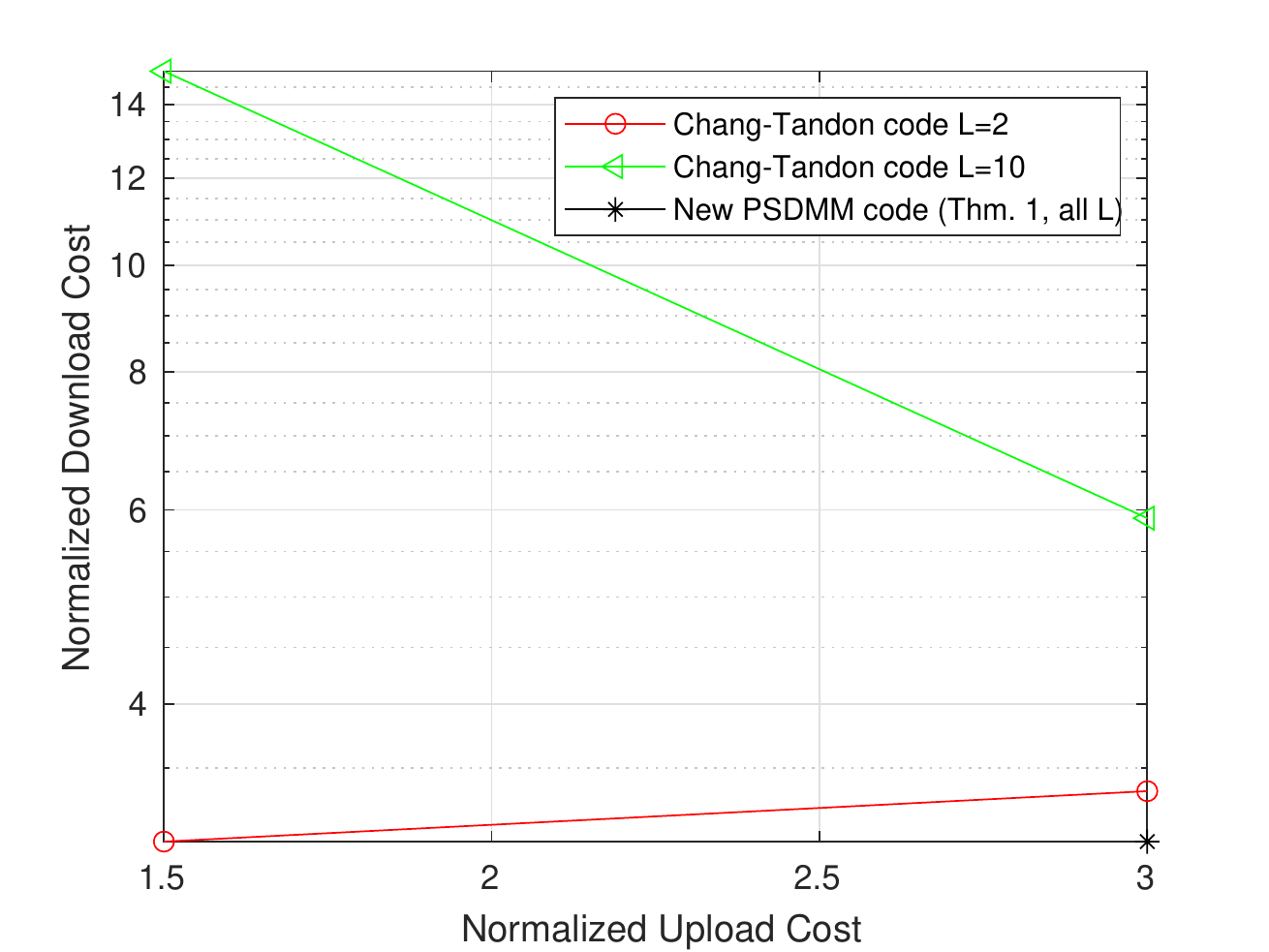}
\end{minipage}
\hspace{18mm}
\begin{minipage}[t]{0.4\textwidth}
\includegraphics[scale=.7]{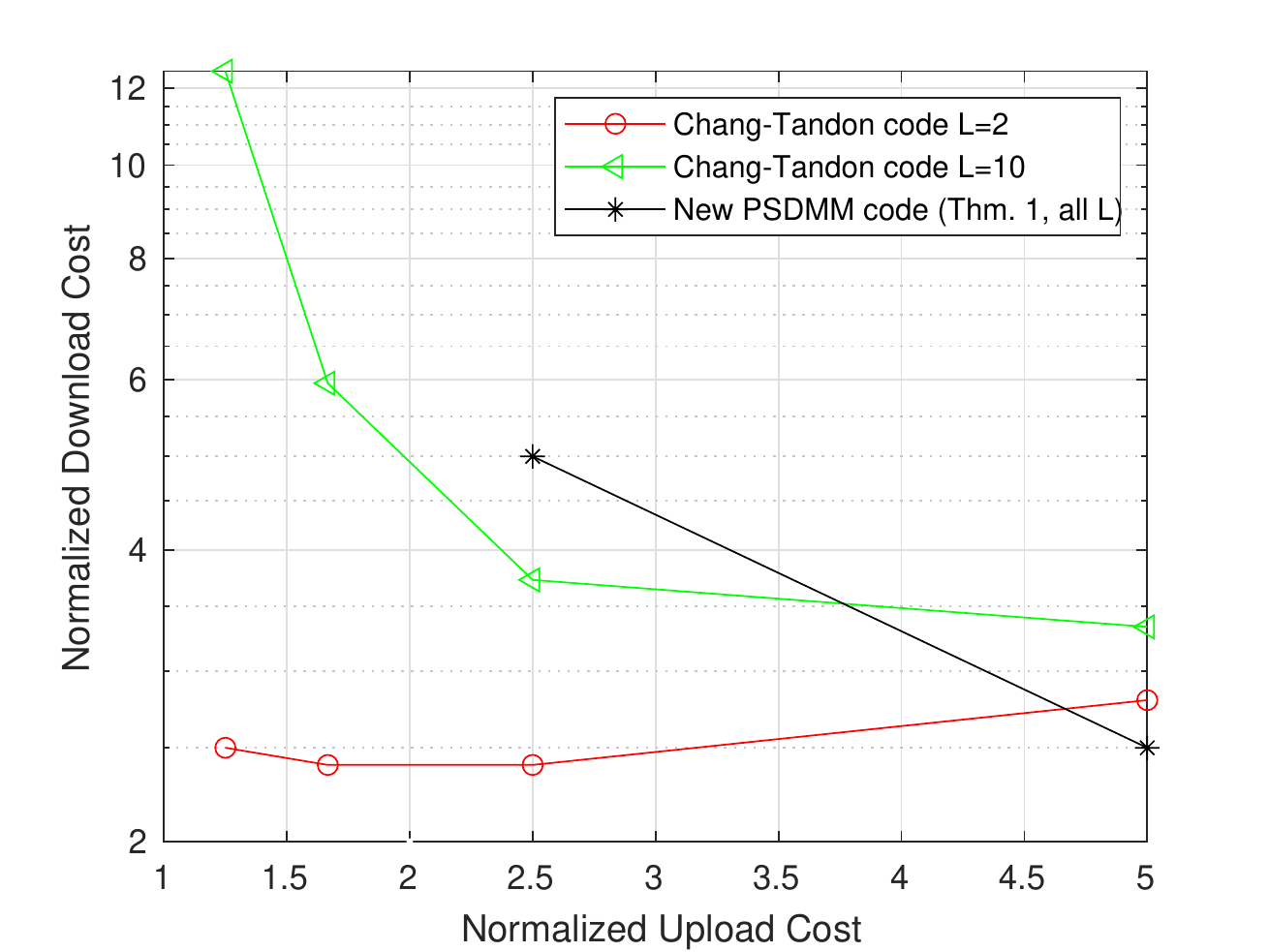}
\end{minipage}
\caption{Comparison  of the normalized upload cost and normalized download cost between the new code $\mathcal{C}_{\rm PSDMM}$ and the Chang--Tandon code under the parameters $N=R_c=3$ (resp. $N=R_c=5$),  where the normalizations of the   upload cost and download cost are over  $|A|$ and $|AB^{(i)}|$, respectively}\label{picture 3-5} 
\end{figure*}

\begin{figure*}[htbp]
\centering
\hspace{-25mm}
\begin{minipage}[t]{0.4\textwidth}
\includegraphics[scale=.7]{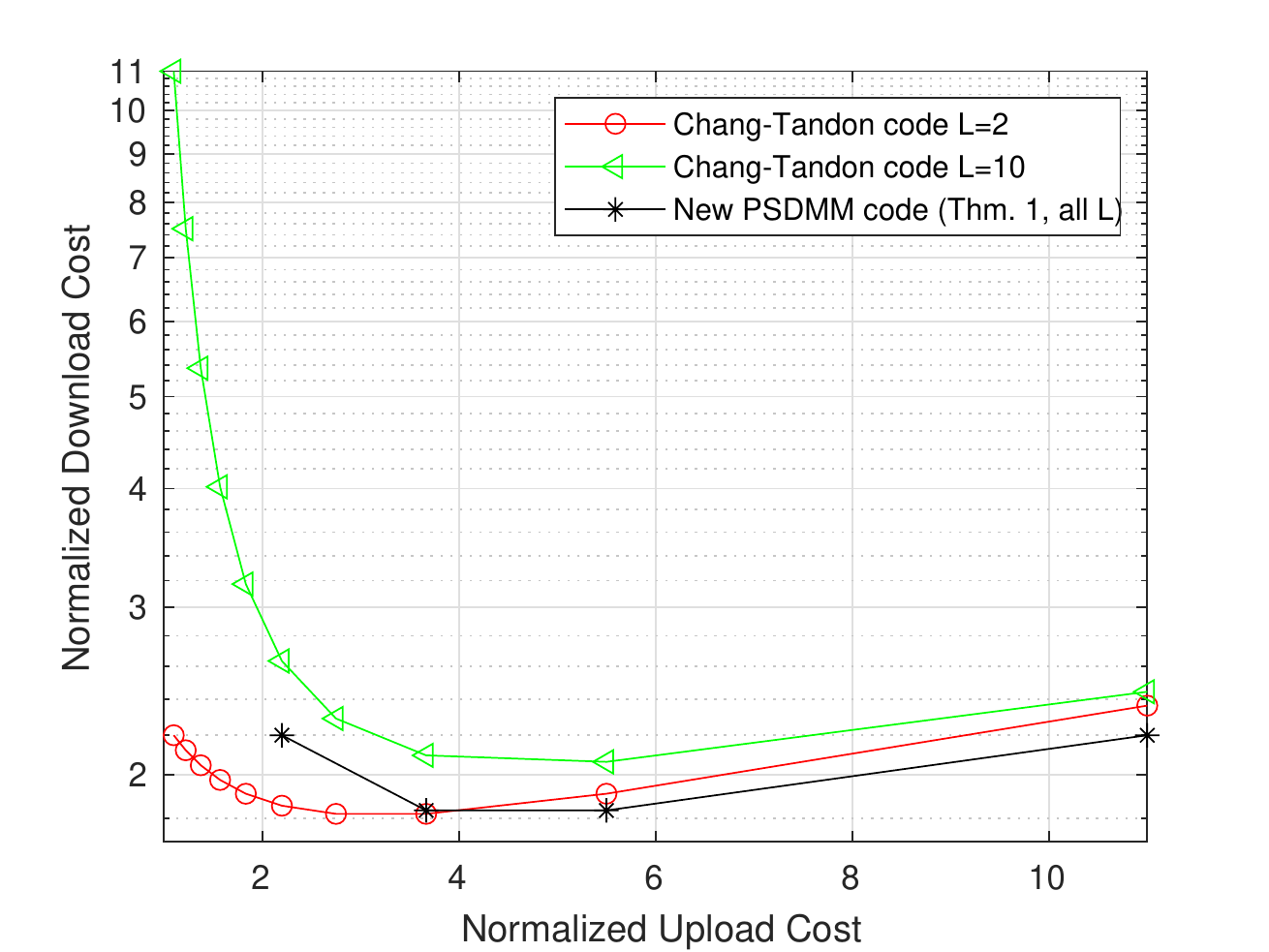}
\end{minipage}
\hspace{18mm}
\begin{minipage}[t]{0.4\textwidth}
\includegraphics[scale=.7]{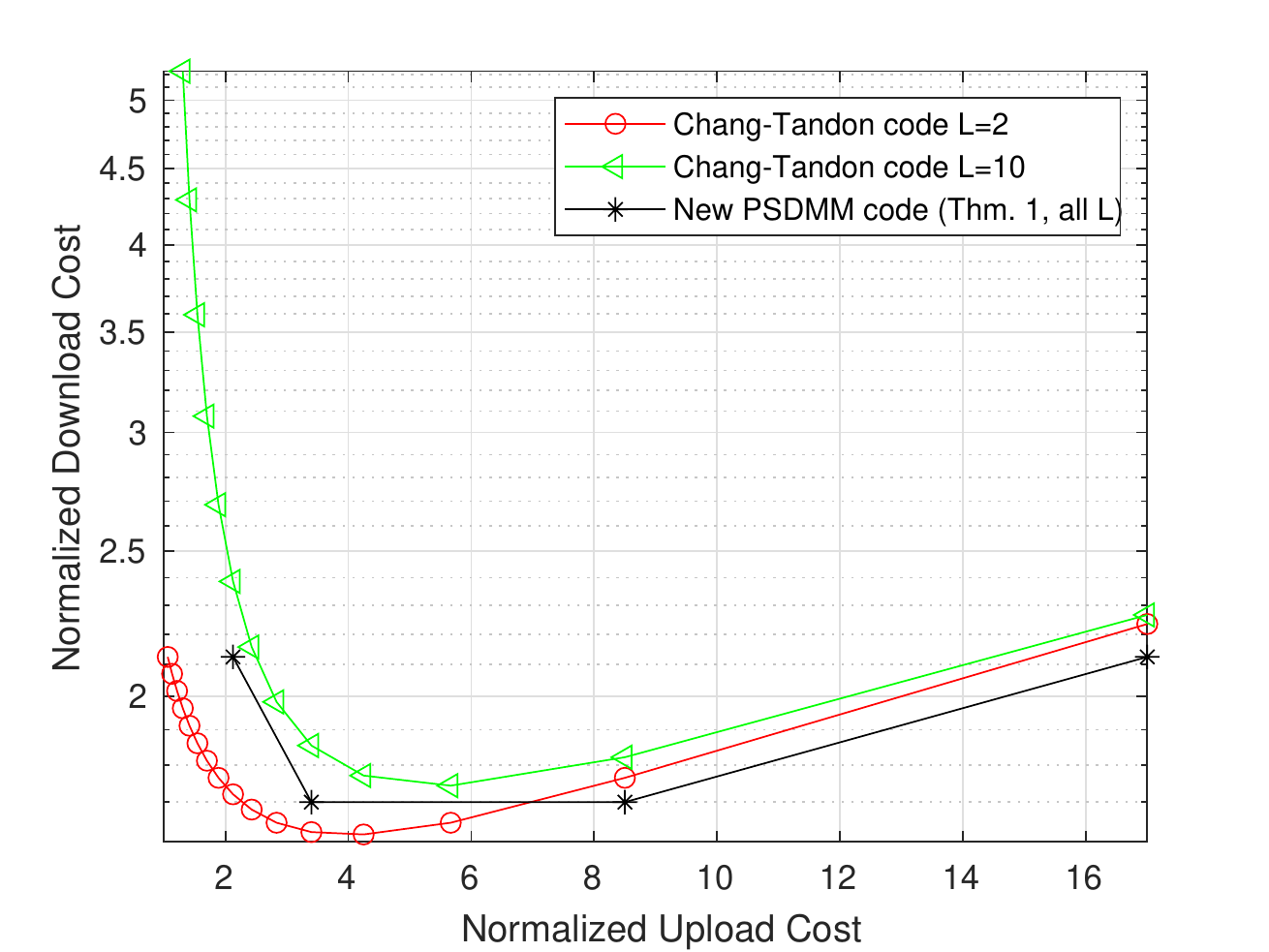}
\end{minipage}
\caption{Comparison  of the normalized upload cost and normalized download cost between the new PSDMM code $\mathcal{C}_{\rm PSDMM}$ and the Chang--Tandon code under the parameters $N=R_c=11$ (resp. $N=R_c=17$)}\label{picture 17}
\end{figure*}

\begin{Proposition}\label{Thm_assignment_PSDMM}
Conditions  (i)--(iii) of Proposition \ref{Thm_PSDMM} can be satisfied if $R_c=pmn+pm+n$,
\begin{equation*}
\alpha_{k,j}=j+kp+1, \beta_{j,k'}=pm-j+k'(pm+1)
\end{equation*}
for $k\in [0, m)$, $j\in [0, p)$, and $k'\in [0, n)$,  and $\gamma=0$.
\end{Proposition}

Before the proof of the above proposition, below we provide some intuitions and insights about the assignment method: 
\begin{itemize}
    \item We wish the degrees of the terms of $f(x)$ and $g(x)$ (or $g_t(x)$) are as small as possible to possibly lower the degree of $h(x)$. W.O.L.G., for $f(x)$, we first set the degree of the term related to the random matrix $Z$ to be $0$, \textit{i.e.,} $\gamma=0$, and then set the degrees of the rest $pm$ terms of $f(x)$ from $1$ to $pm$ for simplicity, \textit{i.e.,}  $\alpha_{k,j}=j+kp+1$.
\item Given $\alpha_{k,j}=j+kp+1$, Proposition  \ref{Thm_PSDMM}-(i) implies that $\beta_{i,k'}=\beta_{0,k'}-i$ for $i\in [1, p)$, \textit{i.e.,} $\beta_{0,k'}$ is the largest integer among $\beta_{0,k'}, \beta_{1,k'},\ldots,\beta_{p-1,k'}$. Thus it is suffice to fix the values of $\beta_{0,0},\beta_{0,1},\ldots,\beta_{0,n-1}$. Note that Proposition  \ref{Thm_PSDMM}-(ii) requires $U\cap I=\emptyset$, which implies that $\alpha_{0,0}+\beta_{0,0}\not \in \{\alpha_{k,j}|0\le k<m, 0\le j<p\}$, it can be fulfilled if $$\alpha_{0,0}+\beta_{0,0}>\max \{\alpha_{k,j}|0\le k<m, 0\le j<p\}=pm,$$ \textit{i.e.,} $\beta_{0,0}\ge pm$. W.O.L.G., we set $\beta_{0,0}=pm$. 
\item $\beta_{0,k'}$, $k'\in [1, n)$ can be determined similarly according to  Proposition  \ref{Thm_PSDMM}-(i) and (ii).
\end{itemize}

\begin{proof}
From Proposition \ref{Thm_assignment_PSDMM}, we have 
\begin{equation}\label{Eqn_alpha+beta}
\alpha_{k,j}+\beta_{j,k'}\equiv (k'+1)(pm+1)+kp
\end{equation}
for $j\in [0, p)$, $k\in [0, m)$ and $k'\in [0,n)$, thus Proposition \ref{Thm_PSDMM}-(i) holds, and then we easily have
\begin{equation*}
|U|\hspace{-.2mm}=\hspace{-.2mm}|\{(k'+1)(pm+1)+kp|0\le k<m, 0\le k'<n\}|\hspace{-.2mm}=\hspace{-.2mm}mn.    
\end{equation*}
Now let us prove $U\cap I=\emptyset$. 
\begin{itemize}
  \item [1)]
When $j,j'\in [0,p)$ with $j\ne j'$, \textit{i.e.,} $0<|j-j'|\le p-1$, it is easy to see that the element $$\alpha_{k,j}+\beta_{j',k'}=(k'+1)(pm+1)+kp+(j-j')$$ in $I$ is not in the set $U$ for all $k\in [0, m)$, $k'\in [0,n)$, and $j,j'\in [0,p)$ with $j\ne j'$;
  \item [2)] By Proposition \ref{Thm_assignment_PSDMM}, it is also easy to see $$\gamma+\beta_{j',k'}=\beta_{j',k'}=(k'+1)(pm+1)-j'-1\not \in U;$$
  \item [3)] Since $\min U=pm+1$, $\alpha_{k,j}\le pm$, and $\gamma=0$, thus $\gamma\not \in U$ and $\alpha_{k,j}\not \in U$ for all $k\in [0, m)$ and $j\in [0, p)$.
\end{itemize}
Combining 1)-3), we thus proved $U\cap I=\emptyset$, \textit{i.e.,} Proposition \ref{Thm_PSDMM}-(ii) holds. It is easy to check that $$\max(U\cup I)=\alpha_{m-1,p-1}+\beta_{0,n-1}=pmn+pm+n-1,$$ thus Proposition \ref{Thm_PSDMM}-(iii) holds.
\end{proof}

By Propositions \ref{Thm_PSDMM} and \ref{Thm_assignment_PSDMM}, we can derive the result in Theorem \ref{Thm_C5}.

\begin{Remark}
There may  exist other assignment methods of the exponents of $f(x)$ and $g(x)$ that are better than those in Proposition \ref{Thm_assignment_PSDMM}.  Finding the optimal assignment of the exponents of $f(x)$ and $g(x)$  is an interesting problem.
\end{Remark}

\subsection{Comparison}\label{sec:PSDMM_comp}
Next, we provide  comparisons of some key parameters of the proposed PSDMM code $\mathcal{C}_{\rm PSDMM}$ and some previous codes. Table \ref{comp_PSDMM} illustrates the first comparison. To this end, note that Chang--Tandon code \cite{chang2019upload} cannot mitigate stragglers and the number of servers can be an arbitrary integer larger than $m$ but should be prefixed.  Table \ref{comp_PSDMM_nor} gives a (fairer) comparison of the normalized parameters of the proposed PSDMM code and Chang--Tandon code, where the \textit{normalized upload cost} is the ratio of upload cost to  $|A|$ and the \textit{normalized download cost} is the ratio of the download cost to   $|AB^{(i)}|$.  Figures  \ref{picture 3-5} and \ref{picture 17}    visualizes the comparison in Table \ref{comp_PSDMM_nor} by assuming $N=R_c$ in the new PSDMM code so that it is comparable with the Chang--Tandon code \cite{chang2019upload}.

\begin{Remark}\label{Remark_BC}
Although it is well known that $R(p,m,n)<pmn$, the value of $R(p,m,n)$ is not yet known in general, even for some small parameters. Finding out the value of $R(p,m,n)$ is challenge in general \cite{yu2020straggler,haastad1990tensor}. In the case that  the value of $R(p,m,n)$ is not yet known, the constructions in \cite{yu2020entangled} are then built on  known upper bound constructions, which are scattered in the literature. Extensive background in this direction is available in \cite{blaser2013fast,landsberg2017geometry,pan2018fast,burgisser2013algebraic}, and some known upper bounds for $R(p,m,n)$ can be found in \cite{smirnov2013bilinear,laderman1976noncommutative,sedoglavic2017non5,sedoglavic2017non7}.
\end{Remark}

As the value of $R(p,m,n)$ in Table \ref{comp_PSDMM} is unknown in general by Remark \ref{Remark_BC}, we show  in Table \ref{Table_com_bi} that the code $\mathcal{C}_{\rm SDMM}$ outperforms the   SDMM  code in \cite{yu2020entangled} in terms of the recovery threshold for some small parameters.

\begin{table}[htbp]
\begin{center}
\caption{A comparison of the recovery thresholds between  the PSDMM code based on bilinear complexity in \cite{yu2020entangled} and  the new code $\mathcal{C}_{\rm PSDMM}$ for some small $p,m,n$, where $R^*(p, m, n)$ denotes the best upper bound of $R(p, m, n)$ and the values in the last column are rounded to $4$ decimal places }\label{Table_com_bi}
\setlength{\tabcolsep}{2.6pt}
\begin{tabular}{|c|c|c|c|c|c|c|}
\hline $p$ & $m$ & $n$ & $R^*(p, m, n)$ & $R_c^{\mbox{\cite{yu2020entangled} }}$ & $R_c^{\mathcal{C}_{\rm PSDMM}}$ & $\frac{R_c^{\mbox{\cite{yu2020entangled} }}-R_c^{\mathcal{C}_{\rm PSDMM}}}{R_c^{\mbox{\cite{yu2020entangled} }}}$\\
\hline $2$ & $2$ & $2$ & $7$ \cite{strassen1969gaussian}& $15$ & $14$ & $6.67\%$\\
\hline $2$ & $3$ & $3$ & $15$ \cite{burgisser2013algebraic} & $31$ & $27$& $12.9\%$\\
\hline $3$ & $3$ & $3$ & $23$ \cite{laderman1976noncommutative}& $47$ & $39$& $17.02\%$\\
\hline $4$ & $4$ & $4$ & $49$ \cite{strassen1969gaussian}& $99$ & $84$& $15.15\%$\\
\hline $5$ & $5$ & $5$ & $99$ \cite{sedoglavic2017non5}& $199$ & $155$& $22.11\%$\\
\hline $6$ & $6$ & $6$ & $160$ \cite{smirnov2013bilinear}& $321$ & $258$& $19.63\%$\\
\hline $7$ & $7$ & $7$ & $250$ \cite{sedoglavic2017non7}& $501$ & $399$& $20.36\%$\\
\hline $8$ & $8$ & $8$ & $343$ \cite{strassen1969gaussian}& $687$ & $584$& $14.99\%$\\
\hline $9$ & $9$ & $9$ & $520$ \cite{sedoglavic2017non7}& $1041$ & $819$& $21.33\%$\\
  \hline
\end{tabular}
\end{center}
\end{table}

From  Tables \ref{comp_PSDMM} and \ref{comp_PSDMM_nor}, the numerical comparison in Table \ref{Table_com_bi}, and Figures  \ref{picture 3-5} and \ref{picture 17},   we can see that the new PSDMM code $\mathcal{C}_{\rm PSDMM}$ has some advantages over the previous constructions:
\begin{itemize}
\item The new PSDMM code $\mathcal{C}_{\rm PSDMM}$  is  more  general than  Kim--Lee code in  \cite{kim2019privateCL}. Namely, it  can provide a more flexible tradeoff between the upload cost  and    the download cost. In particular, when $p=1$, the  recovery threshold of the new code is one smaller than that of Kim--Lee code in  \cite{kim2019privateCL} under the same upload cost.
\item The recovery threshold and the download cost of  the new PSDMM code $\mathcal{C}_{\rm PSDMM}$ are smaller than those of  Aliasgari \textit{et al.} PSGPD code\footnote{In \cite{aliasgari2020private}, there was a small mistake in the recovery threshold of the PSGPD code. A quick correction could be adding $st-s+2$ to the exponents of each term in the polynomial $F_{B^{(r)}}(z)$ in Eq. (43). In Table \ref{comp_PSDMM}, we give the correct value of the recovery threshold.} in  \cite{aliasgari2020private} under the same upload cost and the PSDMM  code in \cite{yu2020entangled}    for some small parameters, more specially, when 
$R(p,m,n)\ge (pmn+pm+n-1)/2$, see Table \ref{Table_com_bi} for some small parameter regions. 
    
\item The download cost of Chang--Tandon code in \cite{chang2019upload} depends on $L$ and is a  monotonically increasing function w.r.t. $L$ (see the trends in Figures  \ref{picture 3-5} and \ref{picture 17}). Besides, Chang--Tandon PSDMM  code cannot mitigate stragglers and requires a large sub-packetization degree. In contrast, the new  PSDMM code $\mathcal{C}_{\rm PSDMM}$ does not have such defects. Furthermore, Figures  \ref{picture 3-5} and \ref{picture 17}    show that the new PSDMM code  has a smaller normalized download cost for some given normalized upload cost compared with  Chang--Tandon's PSDMM code. 
\end{itemize}

\begin{Remark} Although there are non-Pareto-optimal points in Figures \ref{picture 3-5} and \ref{picture 17} for the new PSDMM code $\mathcal{C}_{\rm PSDMM}$, they are still of interest as the chosen computational complexity (submatrix dimension for the private matrix $A$) dictates the normalized upload cost. A larger submatrix dimension for the private matrix $A$ results in a smaller normalized upload cost, so we may not be able to choose a small normalized upload cost in  resource-constrained situations.
\end{Remark}

\section{Private and Secure  Distributed Matrix Multiplication From MDS-coded Servers}\label{sec:PSDMM2}
In this section, we propose a new PSDMM code $\mathcal{C'}_{\rm PSDMM}$ from MDS-coded servers, \textit{i.e.,}. the matrices $B^{(0)}, B^{(0)}, \cdots, B^{(L-1)}$ are stored across the servers in an MDS-coded form. Note that this new code does not subsume $\mathcal{C}_{\rm PSDMM}$ even if the underlying MDS code is degraded to a repetition code as they are totally two different schemes. Their difference will be elaborated later.

\subsection{A motivating example}\label{sec:Ex_PSDMM_MDS}
Assume that the user possesses a matrix $A\in \mathbf{F}^{t\times s}$ and there are $L=2$ matrices $B^{(0)}, B^{(1)}\in \mathbf{F}^{s\times r}$ stored across the $N$ servers in an MDS-coded form by an $[N, 4]$ MDS code, where $t,s,r$ are even. 
Partition the matrices $A$ and $B^{(i)}$ into block matrices the same as that in Section \ref{sec:psdmm-ex}.

Let $a_{0}, \ldots, a_{N-1}$ be $N$ pairwise distinct elements from $\mathbf{F}$, and suppose server $i$ stores two coded pieces of $B^{(0)}, B^{(1)}$ as
$g_0(a_i),~g_1(a_i)$
for $i\in [0, N)$, where
\begin{equation}\label{Eqn_g_ex_C2}
g_i(x)=B_{0,0}^{(i)}x^{\beta_{0,0}} + B_{0,1}^{(i)}x^{\beta_{0,1}} +
       B_{1,0}^{(i)} x^{\beta_{1,0}}+ B_{1,1}^{(i)}x^{\beta_{1,1}} 
\end{equation}
for $i=0, 1$,
with $\beta_{0,0}, \beta_{0,1}, \beta_{1,0}, \beta_{1,1}$ being some integers such that the matrices $B^{(0)}$ and $B^{(1)}$ can be decoded from any $4$ out of the $N$ servers, that is 
\begin{equation}\label{Eqn_PSDMM_MDS_condition}
 \det\begin{bmatrix}
          a_{i_0}^{\beta_{0,0}} & a_{i_0}^{\beta_{0,1}} &   a_{i_0}^{\beta_{1,0}} &
         a_{i_0}^{\beta_{1,1}}\\  a_{i_1}^{\beta_{0,0}} & a_{i_1}^{\beta_{0,1}} &   a_{i_1}^{\beta_{1,0}} &
         a_{i_1}^{\beta_{1,1}}\\
          a_{i_2}^{\beta_{0,0}} & a_{i_2}^{\beta_{0,1}} &   a_{i_2}^{\beta_{1,0}} &
         a_{i_2}^{\beta_{1,1}}\\
          a_{i_3}^{\beta_{0,0}} & a_{i_3}^{\beta_{0,1}} &   a_{i_3}^{\beta_{1,0}} &
         a_{i_3}^{\beta_{1,1}}
        \end{bmatrix} \ne 0
\end{equation}
for any $\{i_0, i_1, i_2, i_3\}\subset [0, N)$ and $|\{i_0, i_1, i_2, i_3\}|=4$.

Suppose that the user wishes to obtain $AB^{(0)}$. 
Let $Z_0$ be a random matrix over $\mathbf{F}^{\frac{t}{2}\times \frac{s}{2}}$.  Then, define a  polynomial
\begin{equation}\label{Eqn_PS_ex_f_MDS}
 f(x)=A_{0,0}x^{\alpha_{0,0}}+A_{0,1}x^{\alpha_{0,1}}+A_{1,0}x^{\alpha_{1,0}}+A_{1,1}x^{\alpha_{1,1}}+ Z_0x^{\gamma},
  \end{equation}
where $\alpha_{k,j}, \gamma$ are some integers to be specified later.

Let $S$ be a random matrix over $\mathbf{F}^{\frac{t}{2}\times \frac{s}{2}}$.
For every $i\in [0, N)$, the  user first evaluates $f(x)$ at $a_{i}$, $i\in [0, N)$, and then sends the query
\begin{equation}\label{Eqn_query_ex_MDS}
q_i^{(0)}=(f(a_i),S)
\end{equation}
to server $i$. Upon receiving the query $q_i^{(0)}$, server $i$ computes $f(a_{i})g_0(a_{i})+Sg_1(a_{i})$ and then returns the result to the user.

Let $h(x)=f(x)g_0(x)+Sg_1(x)$, the expression is shown in \eqref{Eqn-h-private-ex-MDS} in the next page,
\begin{figure*}
\hrulefill\small
\begin{align}
\nonumber  h(x)
=&\underbrace{(A_{0,0}B_{0,0}^{(0)}x^{\alpha_{0,0}+\beta_{0,0}}+A_{0,1}B_{1,0}^{(0)}x^{\alpha_{0,1}+\beta_{1,0}}) +(A_{0,0}B_{0,1}^{(0)}x^{\alpha_{0,0}+\beta_{0,1}}+A_{0,1}B_{1,1}^{(0)}x^{\alpha_{0,1}+\beta_{1,1}})}_{\mbox{useful\  terms}}\\
\nonumber&+\underbrace{(A_{1,0}B_{0,0}^{(0)}x^{\alpha_{1,0}+\beta_{0,0}}+A_{1,1}B_{1,0}^{(0)}x^{\alpha_{1,1}+\beta_{1,0}}) +(A_{1,0}B_{0,1}^{(0)}x^{\alpha_{1,0}+\beta_{0,1}}+A_{1,1}B_{1,1}^{(0)}x^{\alpha_{1,1}+\beta_{1,1}})}_{\mbox{useful\  terms}}+
 \underbrace{\sum\limits_{j'=0}^{1}\sum\limits_{k'=0}^{1}Z_0B_{j',k'}^{(0)}x^{\gamma+\beta_{j',k'}}}_{\mbox{interference\  terms}}\\
  \label{Eqn-h-private-ex-MDS}
&+ \underbrace{\sum\limits_{k=0}^{1}\sum\limits_{k'=0}^{1}(A_{k,0}B_{1,k'}^{(0)}x^{\alpha_{k,0}+\beta_{1,k'}}+A_{k,1}B_{0,k'}^{(0)}x^{\alpha_{k,1}+\beta_{0,k'}})+ S \left(B_{0,0}^{(1)}x^{\beta_{0,0}}+B_{0,1}^{(1)}x^{\beta_{0,1}}+B_{1,0}^{(1)}x^{\beta_{1,0}}+B_{1,1}^{(1)}x^{\beta_{1,1}}\right)}_{\mbox{interference\  terms}},
\end{align}
\end{figure*}
then the result returned from server $i$ is exactly $h(a_i)$, where $i\in [0, N)$.

The user wishes to obtain the data in \eqref{Eqn_entry_C_ex2} (related to the useful terms in $h(x)$) from any $R_c$ out of  the $N$ evaluations of $h(x)$, which can be fulfilled if the following conditions hold.

\begin{itemize}
  \item [(i)] For $k\in [0, 2)$ and $k'\in [0, 2)$,
\begin{eqnarray*}
  \alpha_{k,0}+\beta_{0,k'} = \alpha_{k,1}+\beta_{1,k'}.
\end{eqnarray*}
  \item [(ii)]For $U=\{\alpha_{k,0}+\beta_{0,k'}|0\le k<2, 0\le k'<2\}$ and
  \begin{eqnarray*}
   I&=&\{\alpha_{k,j}+\beta_{j',k'}|0\le k, k'<2, 0\le j\ne j'<2\}\\
   &&\cup  \{\gamma+\beta_{j',k'}|0\le j',k'<2\}\\
   &&\cup\{\beta_{k,j}|0\le k, j<2\},
  \end{eqnarray*}
  $|U|=4$ and $U\cap I= \emptyset$.
  \item [(iii)] $R_c=\deg(h(x))+1$.
\end{itemize}

The task can be finished because of the similar reason as in Section \ref{sec:psdmm-ex}.
By (ii), we can get $\gamma\ne \alpha_{k,j}$ for $k,j=0, 1$, thus one easily obtains $I(S; A)=0$ as $S$ is random and $I(f(a_i); A)=0$
for any $i\in [0, N)$ and $a_i\in \mathbf{F}$ according to  \eqref{Eqn_PS_ex_f_MDS}, then security is fulfilled. While the privacy condition is met by the definition of the query vector in \eqref{Eqn_query_ex_MDS} for the desired index, the detailed proof  is similar to \cite{kim2019privateCL}.

We provide a concrete exponent assignment  for this example in Table \ref{Table_private_ex_MDS1}. From the given assignment, we see that $\deg(h(x))=11$ and thus $R_c=12$. The upload cost is $N\frac{ts}{2}$ and the download cost is $R_c\frac{tr}{4}$, while each server only need to store $\frac{sr}{2}$ elements from $\mathbf{F}$. 

\begin{table}[htbp]
\begin{center}
\caption{An assignment for exponents of $f(x)$ in \eqref{Eqn_PS_ex_f_MDS} and $g(x)$ in \eqref{Eqn_g_ex_C2}, where $\spadesuit,\heartsuit,\clubsuit,\diamondsuit$ are used to highlight the exponents of the useful terms of $h(x)$ in \eqref{Eqn-h-private-ex-MDS}.}\label{Table_private_ex_MDS1}
\begin{tabular}{|c|c|c|c|c|c|c|c}
\hline $+$ & $\beta_{0,0}=0$ & $\beta_{0,1}=1$ & $\beta_{1,0}=2$& $\beta_{1,1}=3$ \\
\hline $\alpha_{0,0}=4$ & $4\ \ \spadesuit$ & $5\ \ \heartsuit$ & $6$& $7$ \\
\hline $\alpha_{0,1}=2$& $2$ & $3$ & $4\ \ \spadesuit$& $5\ \ \heartsuit$ \\
\hline $\alpha_{1,0}=8$& $8\ \ \clubsuit$ & $9\ \ \diamondsuit$ & $10$& $11$  \\
\hline $\alpha_{1,1}=6$& $6$ & $7$ & $8\ \ \clubsuit$& $9\ \ \diamondsuit$  \\
\hline $\gamma=0$& $0$ & $1$ & $2$& $3$  \\
  \hline
\end{tabular}
\end{center}
\end{table}

\subsection{General construction}
In the following, we propose a general construction for PSDMM from MDS-coded servers according to an $[N, pn]$ MDS code.
Partition the matrices $A$ and $B^{(i)}$ ($i\in [0, L)$) into block matrices the same as that in Section \ref{sec:PSDMM-ge-con}.

Let $a_{0}, \ldots, a_{N-1}$ be $N$ pairwise distinct elements from $\mathbf{F}$, then server $i$ stores the $L$ matrices
\begin{equation*}
g_0(a_i),~g_1(a_i), \ldots, g_{L-1}(a_i)    
\end{equation*}
for $i\in [0, N)$, where
\begin{equation*}
g_i(x)=\sum\limits_{j=0}^{p-1}\sum\limits_{k=0}^{n-1}B_{j,k}^{(i)}x^{\beta_{j,k}}
\end{equation*}
with 
\begin{equation}\label{Eqn_PSDMM_MDS_condition_ge}
 \det\begin{bmatrix}
          a_{i_0}^{\beta_{0,0}} & a_{i_0}^{\beta_{0,1}} &   \cdots &
         a_{i_0}^{\beta_{p-1,n-1}}\\   a_{i_1}^{\beta_{0,0}} & a_{i_1}^{\beta_{0,1}} &   \cdots &
         a_{i_1}^{\beta_{p-1,n-1}}\\ 
          \vdots & \vdots &   \ddots &
         \vdots\\ 
           a_{i_{pn-1}}^{\beta_{0,0}} & a_{i_{pn-1}}^{\beta_{0,1}} &   \cdots &
         a_{i_{pn-1}}^{\beta_{p-1,n-1}}
        \end{bmatrix} \ne 0
\end{equation}
for any $\{i_0, i_1, \ldots, i_{pn-1}\}\subset [0, N)$ and $|\{i_0, i_1, \ldots, i_{pn-1}\}|=pn$.

Suppose that the user wishes to compute $AB^{(\theta)}$ for some $\theta\in [0, L)$.
Let $Z_{\theta}$ be a random matrix over $\mathbf{F}^{\frac{t}{m}\times \frac{s}{p}}$,  and define a polynomial
\begin{equation*}
  f(x) = \sum\limits_{k=0}^{m-1}\sum\limits_{j=0}^{p-1}A_{k,j}x^{\alpha_{k,j}}+ Z_{\theta}x^{\gamma},
  \end{equation*}
where $\alpha_{k,j}, \gamma$ are some integers to be specified later.
Additionally, let $S_0, S_1, \ldots, S_{L-2}$ be $L-1$ random matrices over $\mathbf{F}^{\frac{t}{m}\times \frac{s}{p}}$.

For every $i\in [0, N)$, the  user first evaluates $f(x)$ at $a_{i}$ and then sends the query
\begin{equation}\label{Eqn_query_MDS}
 q_i^{(\theta)}=(S_0,\ldots,S_{\theta -1}, f(a_i), S_{\theta}, \ldots, S_{L-2})
\end{equation}
to server $i$. Let $q_i^{(\theta)}(j)$ be the $j$-th matrix in  $q_i^{(\theta)}$. 

Upon receiving the query, server $i$ computes 
\begin{equation*}
\sum\limits_{j=0}^{L-1}q_i^{(\theta)}(j)g_j(a_i)=f(a_i)g_{\theta}(a_i)+\sum\limits_{j=0,j\ne \theta}^{L-1}S_jg_j(a_i)
\end{equation*}
and then sends the result back to the user.

Let $h(x)=f(x)g_{\theta}(x)+\sum\limits_{j=0,j\ne \theta}^{L-1}S_jg_j(x)$, \textit{i.e.,}
{
\begin{align}
 \nonumber  h(x)=&\sum\limits_{k=0}^{m-1}\sum\limits_{j=0}^{p-1}\sum\limits_{j'=0}^{p-1}\sum\limits_{k'=0}^{n-1}A_{k,j}B_{j',k'}^{(\theta)}x^{\alpha_{k,j}+\beta_{j',k'}}\\
\nonumber&+ \sum\limits_{j'=0}^{p-1}\sum\limits_{k'=0}^{n-1}Z_{\theta}B_{j',k'}^{(\theta)}x^{\gamma+\beta_{j',k'}}\\
\nonumber &+ \sum\limits_{t=0, t\ne \theta}^{L-1}S_t \left(\sum\limits_{j=0}^{p-1}\sum\limits_{k=0}^{n-1}B_{j,k}^{(t)}x^{\beta_{j,k}}\hspace{-.5mm} \right),
\end{align}}then the result returned from server $i$ is exactly $h(a_i)$,
thus we can derive the following result.

\begin{table*}[htbp]
\begin{center}
\caption{A comparison of some key parameters between the proposed PSDMM code $\mathcal{C'}_{\rm PSDMM}$ from MDS-coded servers  (broadcast version of upload in \begin{color}{blue}blue\end{color}) and $\mathcal{C}_{\rm PSDMM}$ from replicated servers in Section \ref{sec:PSDMM}, where we assume that the underlying field size is no less than $N\ge R_c$ for a fair comparison.}
\label{comp_PSDMM_MDS}
\begin{tabular}{|c|c|c|c|c|c|}
\hline  & Storage per server & \begin{color}{blue}(Broadcast)\end{color} upload cost & Download cost & Recovery threshold $R_c$& References \\
\hline $\mathcal{C}_{\rm PSDMM}$ & $Lsr$ & $N\frac{ts}{mp}$ & $R_c\frac{tr}{mn}$ & $pmn+pm+n$ & Theorem \ref{Thm_C5}\\
\hline  $\mathcal{C'}_{\rm PSDMM}$ & $L\frac{sr}{pn}$ & \begin{color}{blue}{$\left((N+L-1)\frac{ts}{mp}\right)$}\end{color}\quad  $LN\frac{ts}{mp}$    & $R_c\frac{tr}{mn}$ & $pmn+pn-1$ & Theorem \ref{Thm_C6}\\
\hline
\end{tabular}
\end{center}
\end{table*}
\begin{Proposition}\label{Thm_PSDMM_MDS}
For PSDMM schemes for MDS-coded servers, the multiplication of   $A$ and $B^{(\theta)}$ can be securely computed with the upload cost  $LN\frac{ts}{mp}$, download cost  $R_c\frac{tr}{mn}$ with $R_c$ being the recovery threshold, and each server stores $L\frac{sr}{pn}$ elements from $\mathbf{F}$
 if the following conditions hold.
\begin{itemize}
  \item [(i)] For $k\in [0, m)$ and $k'\in [0, n)$,
\begin{eqnarray*}
  \alpha_{k,0}+\beta_{0,k'} = \cdots=\alpha_{k,p-1}+\beta_{p-1,k'}.
\end{eqnarray*}
  \item [(ii)]$U=\{\alpha_{k,0}+\beta_{0,k'}|0\le k<m, 0\le k'<n\}$ is disjoint with
{\small   
  \begin{align*}
I&=\{\alpha_{k,j}+\beta_{j',k'}|0\le k<m, 0\le j\ne j'<p, 0\le k'<n\}\\
   &\cup \{\gamma+\beta_{j',k'}|0\le j'<p, 0\le k'<n\}\\&\cup\{\beta_{j,k}|0\le j<p, 0\le k<n\},
  \end{align*}
}
  and $|U|=mn$.
  \item [(iii)] $R_c= \deg(h(x))+1=\max(U\cup I)+1$.
\end{itemize}
\end{Proposition}

\begin{proof}
By (ii), one can deduce that $\alpha_{0,0},\ldots,\alpha_{m-1,p-1},\gamma$ are pairwise distinct, then
\begin{equation*}
  I(\{f(a_i)\}; A)=0
\end{equation*}
for any $i\in [0, N)$ and $a_i\in \mathbf{F}$ as $f(a_i)$ is masked by the random matrix $Z_{\theta}$. Since $S_0, \ldots, S_{L-2}$ are random matrices, we further have
\begin{equation*}
  I(f(a_i),S_0, \ldots, S_{L-2}; A)=0
\end{equation*}\textit{i.e.,} security is guaranteed. The privacy condition follows from   the definition of the query in Eq. \eqref{Eqn_query_MDS} for the desired index $\theta\in [0, L)$, which is random to each server. The detailed proof is similar to \cite{kim2019privateCL}, thus we omit it here.
By (i),   each $C_{k,k'}$ appears in $h(x)$,
 while (ii) guarantees that   each $C_{k,k'}$ is the coefficient of a unique term in $h(x)$, and finally
  (iii) guarantees the decodability.

It is obvious  that the upload cost   is $\sum\limits_{i=0}^{N-1}|q_i^{(\theta)}|=NL\frac{ts}{mp}$ and the download cost is   $R_c\frac{tr}{mn}$. This finishes the proof.
\end{proof}

\begin{Remark} In the above proposition, we can also broadcast the query vector component by component to all the servers at once, except for the coordinate $f(a_i)$, which has to be individually sent to each server. Depending on the relation of the size of $L$ versus $N$, this may imply significant saving in the upload cost. In the case of broadcasting, the upload cost would be $(N+L-1)\frac{ts}{mp}$.
\end{Remark}

In the following, we provide an assignment method for the exponents of $f(x)$ and $g(x)$.

\begin{Proposition}\label{Thm_assignment_PSDMM_MDS}
Conditions  (i)--(iii) of Proposition \ref{Thm_PSDMM_MDS} can be satisfied if $R_c=pmn+pn-1$,
\begin{equation*}
\alpha_{k,j}=(k+1)pn-jn, \beta_{j,k'}=jn+k' 
\end{equation*}
for $k\in [0, m)$, $j\in [0, p)$, $k'\in [0, n)$, and 
$\gamma=0$.
\end{Proposition}

Some intuitions and insights about the assignment method:
\begin{itemize}
    \item If $\beta_{0,0}, \ldots, \beta_{0,n-1}, \ldots, \beta_{p-1,0}, \ldots, \beta_{p-1,n-1}$ are in succession, then the matrix in \eqref{Eqn_PSDMM_MDS_condition_ge} is a Vandermonde matrix and thus \eqref{Eqn_PSDMM_MDS_condition_ge} is satisfied. Thus we set $\beta_{j,k'}=jn+k'$ for $j\in [0, p)$, and $k'\in [0, n)$.

\item We set $\gamma=0$, because we will have  
\begin{eqnarray*}
  &&\{\gamma+\beta_{j',k'}|0\le j'<p, 0\le k'<n\}\\&=&\{\beta_{j,k}|0\le j<p, 0\le k<n\},
\end{eqnarray*}
thus the cardinality of the set $I$ can be reduced by its definition in Proposition \ref{Thm_PSDMM_MDS}-(ii).
\item Given $\beta_{j,k'}=jn+k'$, Proposition  \ref{Thm_PSDMM_MDS}-(i) implies that $\alpha_{k,j}=\alpha_{k,0}-jn$ for $j\in [1, p)$. Since $$\min \{\beta_{j,k}|0\le j<p, 0\le k<n\}=pn-1$$ and $$U\cap \{\beta_{j,k}|0\le j<p, 0\le k<n\}=\emptyset,$$ which can be fulfilled if $\min U\ge pn$. W.O.L.G., we set $\alpha_{0,0}+\beta_{0,0}=pn$, then $\alpha_{0,j}=pn-jn$ for $j\in [1, p)$. 
\item $\alpha_{k,0}$, $k\in [1, n)$ can be determined similarly according to  Proposition  \ref{Thm_PSDMM_MDS}-(i) and (ii).
\end{itemize}

\begin{proof}
With the above intuitions and insights, it is straightforward and easy to check  conditions (i)--(iii) of Proposition \ref{Thm_PSDMM_MDS}, therefore, we omit the proof here.
\end{proof}

{
\begin{Remark}
Note that if the polynomial $h(x)$ contains $N_{dt}$ distinct terms with $N_{dt}\le \deg (h(x))$, then the coefficients of the polynomial $h(x)$ can be possibly retrieved if one gets $N_{dt}$ points on the curve $y=h(x)$ and the corresponding  $N_{dt}$ equations are linear independent, which can be fulfilled
if the scheme is built over a sufficiently large finite field. That is, in order to minimize the recovery threshold, it is also feasible to minimize the number $N_{dt}$ of distinct terms of $h(x)$ instead of its degree if the scheme is over a sufficiently large finite field, then the recovery threshold is $R_c=N_{dt}$. This is indeed the case as in the SDMM problem in \cite{d2020gasp,d2021degree}. 
\end{Remark}

In this paper, we aim to build schemes over a small finite field, \textit{i.e.,} characterized the recovery threshold by $\deg (h(x))$ other than $N_{dt}$. Nevertheless, we still provide an alternative assignment method
for the case that the underlying finite field is sufficiently large. In this case, Proposition \ref{Thm_PSDMM_MDS}-(iii) can be replaced by 
\begin{itemize}
  \item [(iii')] $R_c=N_{dt}$ if the underlying finite field is sufficiently large.
\end{itemize}
In the meanwhile, $\beta_{0,0}, \ldots, \beta_{0,n-1}, \ldots, \beta_{p-1,0}, \ldots, \beta_{p-1,n-1}$ do not need to be in succession as in Proposition \ref{Thm_assignment_PSDMM_MDS}, as \eqref{Eqn_PSDMM_MDS_condition_ge} can be easily satisfied over a sufficiently large finite field. In this case, we can provide a more efficient exponent assignment in terms of the recovery threshold. Before presenting the general assignment, Table \ref{Table_private_ex_MDS2} gives an example of the assignment following the example in Section \ref{sec:Ex_PSDMM_MDS}, which leads to a smaller recovery threshold, \textit{i.e.,} $11$.
\begin{table}[htbp]
\begin{center}
\caption{An alternative assignment for exponents of $f(x)$ in \eqref{Eqn_PS_ex_f_MDS} and $g(x)$ in \eqref{Eqn_g_ex_C2}, where $\spadesuit,\heartsuit,\clubsuit,\diamondsuit$ are used to highlight exponents of the useful terms of $h(x)$ in \eqref{Eqn-h-private-ex-MDS}.}\label{Table_private_ex_MDS2}
\begin{tabular}{|c|c|c|c|c|c|c|c}
\hline $+$ & $\beta_{0,0}=1$ & $\beta_{0,1}=6$ & $\beta_{1,0}=0$& $\beta_{1,1}=5$ \\
\hline $\alpha_{0,0}=1$ & $2\ \ \spadesuit$ & $7\ \ \heartsuit$ & $1$& $6$ \\
\hline $\alpha_{0,1}=2$& $3$ & $8$ & $2\ \ \spadesuit$& $7\ \ \heartsuit$ \\
\hline $\alpha_{1,0}=3$& $4\ \ \clubsuit$ & $9\ \ \diamondsuit$ & $3$& $8$  \\
\hline $\alpha_{1,1}=4$& $5$ & $10$ & $4\ \ \clubsuit$& $9\ \ \diamondsuit$  \\
\hline $\gamma=0$& $1$ & $6$ & $0$& $5$  \\
  \hline
\end{tabular}
\end{center}
\end{table}
\begin{Proposition}\label{Thm_assignment_PSDMM_MDS-large}
Conditions  (i), (ii) of Proposition \ref{Thm_PSDMM_MDS} and (iii') in the above can be satisfied if 
$R_c=pmn+n+p-1$,
\begin{equation*}
\alpha_{k,j}=j+kp+1, \beta_{j,k'}=p-1-j+k'(pm+1) 
\end{equation*}
for $k\in [0, m)$, $j\in [0, p)$, $k'\in [0, n)$, $\gamma=0$, and the underlying finite field is sufficiently large.
\end{Proposition}
}

By Propositions \ref{Thm_PSDMM_MDS}, \ref{Thm_assignment_PSDMM_MDS}, and \ref{Thm_assignment_PSDMM_MDS-large}, we immediate have the result in Theorem \ref{Thm_C6}.

\subsection{Comparison and Difference between   $\mathcal{C}'_{\rm PSDMM}$ and $\mathcal{C}_{\rm PSDMM}$ in Section \ref{sec:PSDMM}}

In this section, we give a comparison of some key parameters between the proposed PSDMM code $\mathcal{C'}_{\rm PSDMM}$ from MDS-coded servers and $\mathcal{C}_{\rm PSDMM}$ from replicated servers in Section \ref{sec:PSDMM}, see Table \ref{comp_PSDMM_MDS}, and illustrate the difference between the two proposed PSDMM codes.
 
From Table \ref{comp_PSDMM_MDS}, we see that compared to $\mathcal{C}_{\rm PSDMM}$ from replicated servers in Section \ref{sec:PSDMM}, the proposed PSDMM code $\mathcal{C'}_{\rm PSDMM}$ from MDS-coded servers requires much less storage per servers, \textit{i.e.,} only a fraction of $\frac{1}{pn}$ as that of $\mathcal{C}_{\rm PSDMM}$, but at the cost of increasing the upload cost.

Note that although MDS codes subsume repetition codes, it does not mean that $\mathcal{C'}_{\rm PSDMM}$ subsume $\mathcal{C}_{\rm PSDMM}$, as they are totally two different schemes. The difference between the two schemes lies in the following aspects:
\begin{itemize}
 \item \textbf{Upload phase:} In $\mathcal{C}_{\rm PSDMM}$, the user sends an encoded piece of $A$ and $L$ field elements (\textit{i.e.,} 
$q_i^{(\theta)}$ in \eqref{Eqn_query}) to each server, whereas in $\mathcal{C'}_{\rm PSDMM}$ from MDS-coded servers, the user sends an encoded piece of $A$ together with other $L-1$ random matrices (\textit{i.e.,} 
$q_i^{(\theta)}$ in \eqref{Eqn_query_MDS}) to each server.
  \item  \textbf{Matrix partitioning:} In $\mathcal{C}_{\rm PSDMM}$ from replicated servers, the partition of the matrices $B^{(0)}, B^{(0)}, \cdots, B^{(L-1)}$ is carried out by each server before it encodes the library after receiving the query, whereas in $\mathcal{C'}_{\rm PSDMM}$ from MDS-coded servers, the partition of the matrices $B^{(0)}, B^{(0)}, \cdots, B^{(L-1)}$ is predetermined as these matrices are stored across the servers in an MDS-coded form in advance.

  \item \textbf{Polynomial evaluation:}  In $\mathcal{C}_{\rm PSDMM}$ from replicated servers, each server should first encode the library by evaluating $L$ encoding polynomials $g_t(x)$ ($t\in [0, L)$) in \eqref{Eqn_gt} at $L$ evaluation points, respectively, and then carry out the matrix multiplication between a $\frac{t}{m}\times \frac{s}{p}$ matrix and a  $\frac{s}{p}\times \frac{r}{n}$ matrix.  Whereas in $\mathcal{C'}_{\rm PSDMM}$ from MDS-coded servers, each server needs to perform $L$ times the matrix multiplication between a $\frac{t}{m}\times \frac{s}{p}$ matrix and a  $\frac{s}{p}\times \frac{r}{n}$ matrix, but does not need to evaluate polynomials.
  \end{itemize}

\section{Conclusion}\label{sec:conclusion}
We considered the problem of PSDMM and proposed two new coding schemes, \textit{i.e.,}  $\mathcal{C}_{\rm PSDMM}$ from replicated servers and $\mathcal{C'}_{\rm PSDMM}$ from MDS-coded servers. The proposed codes have a better performance than state-of-the-art schemes in that $\mathcal{C}_{\rm PSDMM}$ can achieve a smaller recovery threshold and download cost as well as providing a more flexible tradeoff between the upload and download costs, whereas $\mathcal{C'}_{\rm PSDMM}$ can significantly save the storage in the servers. Characterizing the optimal trade-off between the upload and download costs as well as  various  theoretical bounds for PSDMM provide interesting open problems for further study.

\section*{Acknowledgment}
The authors would like to thank the Associate Editor Prof. Rafael Schaefer and the two anonymous reviewers for their valuable suggestions and comments, which have greatly improved the presentation and quality of this paper.	

\bibliographystyle{IEEEtran}
\bibliography{SDMM}

\begin{thebibliography}{10}
\providecommand{\url}[1]{#1}
\csname url@samestyle\endcsname
\providecommand{\newblock}{\relax}
\providecommand{\bibinfo}[2]{#2}
\providecommand{\BIBentrySTDinterwordspacing}{\spaceskip=0pt\relax}
\providecommand{\BIBentryALTinterwordstretchfactor}{4}
\providecommand{\BIBentryALTinterwordspacing}{\spaceskip=\fontdimen2\font plus
\BIBentryALTinterwordstretchfactor\fontdimen3\font minus
  \fontdimen4\font\relax}
\providecommand{\BIBforeignlanguage}[2]{{%
\expandafter\ifx\csname l@#1\endcsname\relax
\typeout{** WARNING: IEEEtran.bst: No hyphenation pattern has been}%
\typeout{** loaded for the language `#1'. Using the pattern for}%
\typeout{** the default language instead.}%
\else
\language=\csname l@#1\endcsname
\fi
#2}}
\providecommand{\BIBdecl}{\relax}
\BIBdecl

\bibitem{li2021improved}
J.~Li and C.~Hollanti, ``Improved private and secure distributed matrix
  multiplication,'' in \emph{2021 IEEE International Symposium on Information
  Theory (ISIT)}.\hskip 1em plus 0.5em minus 0.4em\relax IEEE, 2021, pp. 1--6.

\bibitem{joshi2017efficient}
G.~Joshi, E.~Soljanin, and G.~Wornell, ``Efficient redundancy techniques for
  latency reduction in cloud systems,'' \emph{ACM Transactions on Modeling and
  Performance Evaluation of Computing Systems (TOMPECS)}, vol.~2, no.~2, pp.
  1--30, 2017.

\bibitem{wang2015using}
D.~Wang, G.~Joshi, and G.~Wornell, ``Using straggler replication to reduce
  latency in large-scale parallel computing,'' \emph{ACM SIGMETRICS Performance
  Evaluation Review}, vol.~43, no.~3, pp. 7--11, 2015.

\bibitem{lee2017speeding}
K.~Lee, M.~Lam, R.~Pedarsani, D.~Papailiopoulos, and K.~Ramchandran, ``Speeding
  up distributed machine learning using codes,'' \emph{IEEE Transactions on
  Information Theory}, vol.~64, no.~3, pp. 1514--1529, 2017.

\bibitem{yu2017polynomial}
Q.~Yu, M.~Maddah-Ali, and S.~Avestimehr, ``Polynomial codes: an optimal design
  for high-dimensional coded matrix multiplication,'' in \emph{Advances in
  Neural Information Processing Systems}, 2017, pp. 4403--4413.

\bibitem{dutta2018unified}
S.~Dutta, Z.~Bai, H.~Jeong, T.~M. Low, and P.~Grover, ``A unified coded deep
  neural network training strategy based on generalized polydot codes,'' in
  \emph{2018 IEEE International Symposium on Information Theory (ISIT)}.\hskip
  1em plus 0.5em minus 0.4em\relax IEEE, 2018, pp. 1585--1589.

\bibitem{dutta2019optimal}
S.~Dutta, M.~Fahim, F.~Haddadpour, H.~Jeong, V.~Cadambe, and P.~Grover, ``On
  the optimal recovery threshold of coded matrix multiplication,'' \emph{IEEE
  Transactions on Information Theory}, vol.~66, no.~1, pp. 278--301, 2019.

\bibitem{jia2021cross}
Z.~Jia and S.~A. Jafar, ``Cross subspace alignment codes for coded distributed
  batch computation,'' \emph{IEEE Transactions on Information Theory}, vol.~67,
  no.~5, pp. 2821--2846, 2021.

\bibitem{yu2020straggler}
Q.~Yu, M.~A. Maddah-Ali, and A.~S. Avestimehr, ``Straggler mitigation in
  distributed matrix multiplication: Fundamental limits and optimal coding,''
  \emph{IEEE Transactions on Information Theory}, vol.~66, no.~3, pp.
  1920--1933, 2020.

\bibitem{d2020gasp}
R.~G. D'Oliveira, S.~El~Rouayheb, and D.~Karpuk, ``G{ASP} codes for secure
  distributed matrix multiplication,'' \emph{IEEE Transactions on Information
  Theory}, vol.~66, no.~7, pp. 4038--4050, 2020.

\bibitem{d2021degree}
R.~G. D’Oliveira, S.~El~Rouayheb, D.~Heinlein, and D.~Karpuk, ``Degree tables
  for secure distributed matrix multiplication,'' \emph{IEEE Journal on
  Selected Areas in Information Theory}, vol.~2, no.~3, pp. 907--918, 2021.

\bibitem{chang2018capacity}
W.-T. Chang and R.~Tandon, ``On the capacity of secure distributed matrix
  multiplication,'' in \emph{2018 IEEE Global Communications Conference
  (GLOBECOM)}.\hskip 1em plus 0.5em minus 0.4em\relax IEEE, 2018, pp. 1--6.

\bibitem{kakar2019capacity}
J.~Kakar, S.~Ebadifar, and A.~Sezgin, ``On the capacity and
  straggler-robustness of distributed secure matrix multiplication,''
  \emph{IEEE Access}, vol.~7, pp. 45\,783--45\,799, 2019.

\bibitem{yang2019secure}
H.~Yang and J.~Lee, ``Secure distributed computing with straggling servers
  using polynomial codes,'' \emph{IEEE Transactions on Information Forensics
  and Security}, vol.~14, no.~1, pp. 141--150, 2019.

\bibitem{jia2021capacity}
Z.~Jia and S.~A. Jafar, ``On the capacity of secure distributed batch matrix
  multiplication,'' \emph{IEEE Transactions on Information Theory}, vol.~67,
  no.~11, pp. 7420--7437, 2021.

\bibitem{aliasgari2020private}
M.~Aliasgari, O.~Simeone, and J.~Kliewer, ``Private and secure distributed
  matrix multiplication with flexible communication load,'' \emph{IEEE
  Transactions on Information Forensics and Security}, vol.~15, pp. 2722--2734,
  2020.

\bibitem{yu2020entangled}
Q.~Yu and A.~S. Avestimehr, ``Entangled polynomial codes for secure, private,
  and batch distributed matrix multiplication: Breaking the `cubic' barrier,''
  in \emph{2020 IEEE International Symposium on Information Theory
  (ISIT)}.\hskip 1em plus 0.5em minus 0.4em\relax IEEE, 2020, pp. 245--250.

\bibitem{zhu2021secure}
J.~Zhu and X.~Tang, ``Secure batch matrix multiplication from grouping lagrange
  encoding,'' \emph{IEEE Communications Letters}, vol.~25, no.~4, pp.
  1119--1123, 2021.

\bibitem{zhu2021improved}
J.~Zhu, Q.~Yan, and X.~Tang, ``Improved constructions for secure multi-party
  batch matrix multiplication,'' \emph{IEEE Transactions on Communications},
  vol.~69, no.~11, pp. 7673--7690, 2021.

\bibitem{chor1995private}
B.~Chor, O.~Goldreich, E.~Kushilevitz, and M.~Sudan, ``Private information
  retrieval,'' in \emph{Proceedings of IEEE 36th Annual Foundations of Computer
  Science}.\hskip 1em plus 0.5em minus 0.4em\relax IEEE, 1995, pp. 41--50.

\bibitem{sun2017capacity}
H.~Sun and S.~A. Jafar, ``The capacity of private information retrieval,''
  \emph{IEEE Transactions on Information Theory}, vol.~63, no.~7, pp.
  4075--4088, 2017.

\bibitem{sun2017capacityrobust}
------, ``The capacity of robust private information retrieval with colluding
  databases,'' \emph{IEEE Transactions on Information Theory}, vol.~64, no.~4,
  pp. 2361--2370, 2017.

\bibitem{freij2017private}
R.~Freij-Hollanti, O.~W. Gnilke, C.~Hollanti, and D.~A. Karpuk, ``Private
  information retrieval from coded databases with colluding servers,''
  \emph{SIAM Journal on Applied Algebra and Geometry}, vol.~1, no.~1, pp.
  647--664, 2017.

\bibitem{banawan2018capacity}
K.~Banawan and S.~Ulukus, ``The capacity of private information retrieval from
  coded databases,'' \emph{IEEE Transactions on Information Theory}, vol.~64,
  no.~3, pp. 1945--1956, 2018.

\bibitem{tajeddine2018private}
R.~Tajeddine, O.~W. Gnilke, and S.~El~Rouayheb, ``Private information retrieval
  from {MDS} coded data in distributed storage systems,'' \emph{IEEE
  Transactions on Information Theory}, vol.~64, no.~11, pp. 7081--7093, 2018.

\bibitem{freij2018t}
R.~Freij-Hollanti, O.~W. Gnilke, C.~Hollanti, A.-L. Horlemann-Trautmann,
  D.~Karpuk, and I.~Kubjas, ``$ t $-private information retrieval schemes using
  transitive codes,'' \emph{IEEE Transactions on Information Theory}, vol.~65,
  no.~4, pp. 2107--2118, 2018.

\bibitem{kumar2019achieving}
S.~Kumar, H.-Y. Lin, E.~Rosnes, and A.~G. i~Amat, ``Achieving maximum distance
  separable private information retrieval capacity with linear codes,''
  \emph{IEEE Transactions on Information Theory}, vol.~65, no.~7, pp.
  4243--4273, 2019.

\bibitem{zhu2019new}
J.~Zhu, Q.~Yan, C.~Qi, and X.~Tang, ``A new capacity-achieving private
  information retrieval scheme with (almost) optimal file length for coded
  servers,'' \emph{IEEE Transactions on Information Forensics and Security},
  vol.~15, pp. 1248--1260, 2019.

\bibitem{chang2019upload}
W.-T. Chang and R.~Tandon, ``On the upload versus download cost for secure and
  private matrix multiplication,'' in \emph{2019 IEEE Information Theory
  Workshop (ITW)}.\hskip 1em plus 0.5em minus 0.4em\relax IEEE, 2019, pp. 1--5.

\bibitem{wang2019symmetric}
Q.~Wang and M.~Skoglund, ``Symmetric private information retrieval from {MDS}
  coded distributed storage with non-colluding and colluding servers,''
  \emph{IEEE Transactions on Information Theory}, vol.~65, no.~8, pp.
  5160--5175, 2019.

\bibitem{d2019one}
R.~G. D'Oliveira and S.~El~Rouayheb, ``One-shot {PIR}: Refinement and
  lifting,'' \emph{IEEE Transactions on Information Theory}, vol.~66, no.~4,
  pp. 2443--2455, 2019.

\bibitem{zhou2020capacity}
R.~Zhou, C.~Tian, H.~Sun, and T.~Liu, ``Capacity-achieving private information
  retrieval codes from {MDS}-coded databases with minimum message size,''
  \emph{IEEE Transactions on Information Theory}, vol.~66, no.~8, pp.
  4904--4916, 2020.

\bibitem{jia2020x}
Z.~Jia and S.~A. Jafar, ``X-secure t-private information retrieval from {MDS}
  coded storage with byzantine and unresponsive servers,'' \emph{IEEE
  Transactions on Information Theory}, vol.~66, no.~12, pp. 7427--7438, 2020.

\bibitem{li2020towards}
J.~Li, D.~Karpuk, and C.~Hollanti, ``Towards practical private information
  retrieval from {MDS} array codes,'' \emph{IEEE Transactions on
  Communications}, vol.~68, no.~6, pp. 3415--3425, 2020.

\bibitem{holzbaur2021toward}
L.~Holzbaur, R.~Freij-Hollanti, J.~Li, and C.~Hollanti, ``Toward the capacity
  of private information retrieval from coded and colluding servers,''
  \emph{IEEE Transactions on Information Theory}, vol.~68, no.~1, pp. 517--537,
  2021.

\bibitem{kim2019privateCL}
M.~Kim and J.~Lee, ``Private secure coded computation,'' \emph{IEEE
  Communications Letters}, vol.~23, no.~11, pp. 1918--1921, 2019.

\bibitem{blaser2013fast}
M.~Bl{\"a}ser, ``Fast matrix multiplication,'' \emph{Theory of Computing}, pp.
  1--60, 2013.

\bibitem{stoer2013introduction}
J.~Stoer and R.~Bulirsch, \emph{Introduction to numerical analysis}.\hskip 1em
  plus 0.5em minus 0.4em\relax Springer Science \& Business Media, 2013,
  vol.~12.

\bibitem{haastad1990tensor}
J.~H{\aa}stad, ``Tensor rank is np-complete,'' \emph{Journal of algorithms
  (Print)}, vol.~11, no.~4, pp. 644--654, 1990.

\bibitem{landsberg2017geometry}
J.~M. Landsberg, \emph{Geometry and complexity theory}.\hskip 1em plus 0.5em
  minus 0.4em\relax Cambridge University Press, 2017, vol. 169.

\bibitem{pan2018fast}
V.~Y. Pan, ``Fast feasible and unfeasible matrix multiplication,'' \emph{arXiv
  preprint arXiv:1804.04102}, 2018.

\bibitem{burgisser2013algebraic}
P.~B{\"u}rgisser, M.~Clausen, and M.~A. Shokrollahi, \emph{Algebraic complexity
  theory}.\hskip 1em plus 0.5em minus 0.4em\relax Springer Science \& Business
  Media, 2013, vol. 315.

\bibitem{smirnov2013bilinear}
A.~V. Smirnov, ``The bilinear complexity and practical algorithms for matrix
  multiplication,'' \emph{Computational Mathematics and Mathematical Physics},
  vol.~53, no.~12, pp. 1781--1795, 2013.

\bibitem{laderman1976noncommutative}
J.~D. Laderman, ``A noncommutative algorithm for multiplying 3$\times$3
  matrices using 23 multiplications,'' \emph{Bulletin of the American
  Mathematical Society}, vol.~82, no.~1, pp. 126--128, 1976.

\bibitem{sedoglavic2017non5}
A.~Sedoglavic, ``A non-commutative algorithm for multiplying 5$\times$5
  matrices using 99 multiplications,'' \emph{arXiv preprint arXiv:1707.06860},
  2017.

\bibitem{sedoglavic2017non7}
------, ``A non-commutative algorithm for multiplying (7$\times $7) matrices
  using 250 multiplications,'' \emph{arXiv preprint arXiv:1712.07935}, 2017.

\bibitem{strassen1969gaussian}
V.~Strassen, ``Gaussian elimination is not optimal,'' \emph{numerical
  mathematics}, vol.~13, no.~4, pp. 354--356, 1969.

\end{thebibliography}

\begin{IEEEbiography}[{\includegraphics[width=1in,height=1.25in,clip,keepaspectratio]{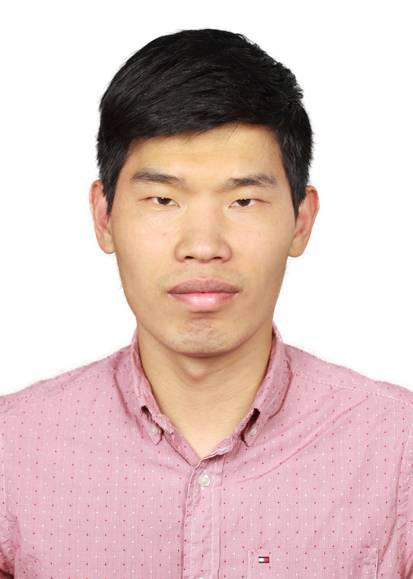}}]
{Jie Li}(Member, IEEE)    
received the B.S. and M.S. degrees in mathematics from Hubei University, Wuhan, China, in 2009 and 2012, respectively, and received the Ph.D. degree from the department of communication engineering, Southwest Jiaotong University, Chengdu, China, in 2017. From  2015 to 2016, he was a visiting Ph.D. student with the Department of Electrical Engineering and Computer Science, The University of Tennessee at Knoxville, TN, USA.  From   2017 to   2019, he was a postdoctoral researcher with the Department of Mathematics, Hubei University, Wuhan, China. From  2019 to 2021, he was a postdoctoral researcher with the Department of Mathematics and Systems Analysis, Aalto University, Finland. He is currently a senior researcher with the Theory Lab, Huawei Tech. Investment Co., Limited, Hong Kong SAR, China. His research interests include distributed computing, private information retrieval, coding for distributed storage, and sequence design.

Dr. Li received the IEEE Jack Keil Wolf ISIT Student Paper Award in 2017.
\end{IEEEbiography}

\begin{IEEEbiography}[{\includegraphics[width=1in,height=1.25in,clip,keepaspectratio]{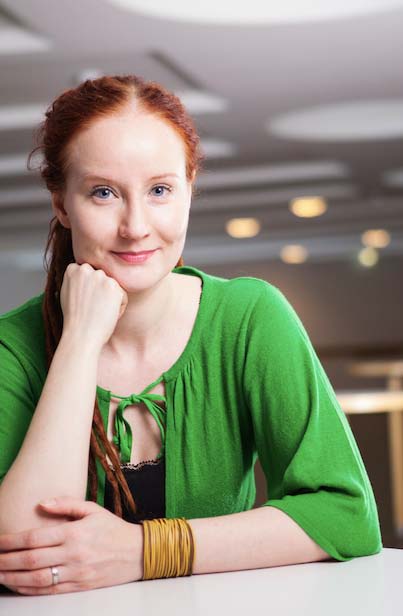}}]
{Camilla Hollanti}(Member, IEEE)     received the M.Sc. and Ph.D. degrees from the University of Turku, Finland, in 2003 and 2009, respectively, both in pure mathematics. Her research interests lie within applications of algebraic number theory to wireless communications and physical layer security, as well as in combinatorial and coding theoretic methods related to distributed storage systems and private information retrieval.

For 2004-2011 Hollanti was with the University of Turku. She joined the University of Tampere as  Lecturer for the academic year 2009-2010. Since 2011, she has been with the Department of Mathematics and Systems Analysis at Aalto University, Finland, where she currently works as Full Professor and Vice Head, and leads a research group in Algebra, Number Theory, and Applications. During 2017-2020, Hollanti was affiliated with the Institute of Advanced Studies at the Technical University of Munich, where she held a three-year Hans Fischer Fellowship, funded by the German Excellence Initiative and the EU 7th Framework Programme.

Hollanti is currently an editor of the AIMS Journal on Advances in Mathematics of Communications,  SIAM Journal on Applied Algebra and Geometry, and IEEE Transactions on Information Theory. She is a recipient of several grants, including six Academy of Finland grants. In 2014, she received the World Cultural Council Special Recognition Award for young researchers. In 2017, the Finnish Academy of Science and Letters awarded her the V\"ais\"al\"a Prize in Mathematics. For 2020-2022, Hollanti is serving as a member of the Board of Governors of the IEEE Information Theory Society, and is one of the General Chairs of IEEE ISIT 2022.
\end{IEEEbiography}
\end{document}